\documentclass[11pt,letterpaper]{article}

\usepackage[margin=1in]{geometry}

\usepackage[T1]{fontenc}
\usepackage[utf8]{inputenc}
\usepackage{lmodern}

\usepackage{amsmath,amssymb,amsfonts,amsthm,mathtools}
\usepackage{bm}
\usepackage{centernot}

\usepackage{algorithm}
\usepackage{algpseudocode}

\usepackage{graphicx}
\usepackage{subcaption}
\usepackage{booktabs}
\usepackage{tikz}

\usepackage{xspace}
\usepackage[super]{nth}
\usepackage{xcolor}
\usepackage{microtype}
\usepackage{url}
\usepackage{array}
\usepackage{cite}
\usepackage{amssymb}
\usepackage{authblk}

\usepackage[hidelinks]{hyperref}
\usepackage[nameinlink,noabbrev]{cleveref}

\usepackage[disable]{todonotes}
\theoremstyle{plain}
\newtheorem{theorem}{Theorem}[section]
\newtheorem{lemma}[theorem]{Lemma}

\newtheorem{corollary}[theorem]{Corollary}

\theoremstyle{definition}
\newtheorem{definition}[theorem]{Definition}

\newtheorem{remark}[theorem]{Remark}

\let\emptyset\varnothing

\newcommand{\gcc}[2]{\Pi_{\text{GC}}}

\newcommand{\safeMEB}{\ensuremath{\mathrm{SafeMEB}}\xspace}
\newcommand{\MEB}{\ensuremath{\mathrm{MEB}}\xspace}
\newcommand{\MEBS}{\ensuremath{\mathrm{MEBs}}\xspace}

\newcommand{\gatherr}{\ensuremath{\mathrm{GTHR}}\xspace}

\title{Fast Multidimensional Approximate Agreement with \\ Optimal Resilience Using Ball Validity}
\author[1]{Tijana Milentijević}
\author[1,2]{Stefan Schmid}

\affil[1]{TU Berlin}
\affil[2]{Weizenbaum Institute}

\date{}
\begin{document}

% Anonymous title page.
% Do not include author names, affiliations, emails, acknowledgements, or funding.

\maketitle

\begin{abstract}
Multidimensional approximate agreement is a fundamental task in distributed computing. It requires $n$ processes with inputs in $\mathbb{R}^d$ to output vectors that are close to each other despite up to $t$ Byzantine faults being present. For the widespread convex validity the output must lie in the convex hull of the correct vectors. Although this guarantee is strong in theory, it becomes prohibitive in high-dimensional scenarios, because the required resilience threshold grows with the dimension.
We therefore study the approximate agreement problem under Minimum Enclosing Ball (\MEB) validity, which has so far been considered for vector consensus in the literature. 
It requires every output to lie within the radius of the minimum enclosing ball of the correct inputs. We also present the $c$-\MEB relaxed validity, where the agreement is within the \MEB with its radius scaled by $c$.

Our first contribution is the Adaptive $\MEB$ Contraction Algorithm, which is coordinate-free and fast: it contracts the correct \MEB radius by $1/\sqrt2\approx 0.707$ per round in the synchronous model with $n>(d+1)t$, while satisfying $\sqrt{2}$-\MEB validity. We also provide an example which shows that this contraction factor cannot be improved. 
%For comparison, the previous fastest algorithm contracts the diameter of the correct values by $\sqrt{7/8}\approx0.935$.

Our main technical contribution is a new dimension-free ball inflation theorem. We prove that if every $\beta$ balls of a finite family of Euclidean balls have a common point, then inflating every radius by $\sqrt{\beta/(\beta-1)}$ makes the whole family intersect. Applying the theorem to the candidate balls defining the local \MEB-safe areas results in a synchronous algorithm with resilience $n>3t$, contraction factor $\sqrt3/2$ and $\sqrt6$-$\MEB$ validity. To the best of our knowledge, this is the first multidimensional approximate agreement algorithm with optimal resilience $n>3t$ that satisfies constant $c$-\MEB validity and contracts by a factor independent of the dimension.

We further extend the approach to the asynchronous setting, using the Gather protocol. We obtain an algorithm that achieves resilience $n>(d+2)t$, has contraction factor $\sqrt{2/3}$ and $\sqrt6$-$\MEB$
validity, and its inflated variant with resilience $n>4t$, contraction factor $\sqrt{15}/4$ and $2\sqrt{10}$-$\MEB$ validity. Finally, we compare our guarantees with existing approximate agreement and contraction algorithms, including Minimum-Diameter Averaging (MDA), for which we derive MEB-validity guarantees. Our algorithms achieve strictly better resilience while providing substantially stronger MEB-validity guarantees than MDA.
\end{abstract}

\section{Introduction}
Agreement is an important task in distributed systems because different components often need to make consistent decisions. In many applications, the processes agree on numerical or geometric values. Examples include sensor measurements~\cite{bandarupalli2024sensorbft}, clock corrections~\cite{lenzen2022optimal}, gradients in distributed learning~\cite{10.5555/3540261.3542179, spaa-geom-median, asia-ccs-coord-median,su2016fault}. In such settings, insisting on exact agreement is, however, often unnecessary, since it is enough that the outputs of the correct processes are sufficiently close to each other. This relaxation is referred to as \emph{approximate agreement}, first introduced in~\cite{10.1145/5925.5931}.

In approximate agreement, each of the $n$ processes starts with an input value, and up to $t$ processes may be Byzantine, meaning that they can behave arbitrarily and omit messages. The goal of approximate agreement is that all correct processes eventually output values that are within distance $\varepsilon$ from each other. At the same time, the output must satisfy a validity condition, i.e. a geometric guarantee which prevents Byzantine processes from pulling the output arbitrarily far away from the values proposed by correct processes.

In this work, we study the multidimensional setting, where the inputs are vectors in $\mathbb{R}^d$. The standard validity condition for multidimensional approximate agreement is \emph{convex validity}: every correct output must lie in the convex hull of the correct input vectors. Convex validity provides a strong theoretical guarantee, but it has a fundamental drawback in higher dimension. Its resilience $n>(d+1)t$ in the synchronous and $n>(d+2)t$ in the asynchronous setting depends on the dimension $d$. Thus, when the dimension is large compared to the number of processes $n$, convex validity severely limits the number of Byzantine faults that can be tolerated.

This motivates the following question:
\begin{center}
\textit{Can we obtain multidimensional approximate agreement with dimension-free resilience while preserving a constant geometric validity guarantee?} 
\end{center}
We answer this question affirmatively using minimum enclosing balls. Instead of requiring outputs to remain in the convex hull of the correct inputs, we require them to remain close to the minimum enclosing ball of the correct inputs. This gives a weaker validity condition, however it allows us to achieve resilience independent of dimension $d$. 

Our goal is not only to achieve a reasonable validity notion, but also to obtain fast convergence. Since approximate agreement is iterative, in every round, each correct process computes a new value, and these new values become the inputs to the next round. Hence, agreement requires a contraction argument showing that the region containing all correct values shrinks over time. We measure this shrinkage by the radius of the minimum enclosing ball of the correct values and design algorithms that shrink this radius directly rather than shrinking one coordinate at a time. This makes the contraction factor independent of the dimension.

\subsection{Our Contributions}
We study multidimensional approximate agreement under Byzantine faults through the lens of minimum enclosing balls. Instead of requiring all correct outputs to remain in the convex hull of the correct inputs, we use the weaker but still geometric notion of $c$-\MEB validity, i.e. all correct outputs must remain within a factor $c$ of the minimum enclosing ball of the initial correct values. This relaxation allows us to go beyond the $n>(d+1)t$ resilience barrier of convex validity while preserving coordinate-free convergence guarantees.

Our first contribution is the Adaptive \MEB Contraction Algorithm, a coordinate-free update rule that contracts the radius of the minimum enclosing ball of the correct values. In the synchronous setting with $n>(d+1)t$, the algorithm achieves contraction factor $1/\sqrt{2}$ for $\alpha=1$ and satisfies $\sqrt{2}$-\MEB validity.
More generally, the parameter $\alpha$ provides a tradeoff between local computation and contraction quality. The case $\alpha=1$ gives the strongest contraction guarantee in this paper. For comparison, MidExtremes, contracts the diameter of correct values by $\sqrt{7/8}\approx 0.935$, whereas our algorithm contracts the radius of the minimum enclosing ball around correct values by $1/\sqrt{2}\approx 0.707$. We also give an example showing that the $1/\sqrt{2}$ factor cannot be improved within our contraction analysis.

Our main contribution is a dimension-free ball inflation theorem. We prove that if every subfamily of at most $\beta$ Euclidean balls has a common point, then inflating every radius by a factor of $\sqrt{\beta/(\beta-1)}$ guarantees that the entire family intersects. This result lets us apply the same contraction idea even when the candidate $\MEB$ balls do not intersect necessarily, that is the case with $n\leq (d+1)t$. In the synchronous setting, this gives an inflated algorithm with optimal resilience $n>3t$, contraction factor $\sqrt{3}/2$ for $\alpha=1$, and $\sqrt{6}$-\MEB validity. To the best of our knowledge this is the first multidimensional approximate agreement algorithm achieving the optimal resilience $n>3t$ with coordinate-free contraction, while still satisfying a meaningful geometric validity condition. The threshold $n>3t$ is optimal for Byzantine agreement.

%and matches  resilience matches matches the resilience achieved in the one-dimensional case by Abraham et al.~\cite{optimal-resilience-asynchronous}.

Our next contribution extends the approach to the asynchronous setting using the Gather protocol. Without inflation, we obtain an asynchronous algorithm for $n>(d+2)t$ with contraction factor $\sqrt{2/3}$ and $\sqrt{6}$-\MEB validity for $\alpha=1$. With inflation, we obtain a dimension-free asynchronous algorithm for $n>4t$ with contraction factor $\sqrt{15}/4$ and $2\sqrt{10}$-\MEB validity.

Finally, we compare our guarantees with prior multidimensional approximate agreement algorithms. The algorithms contract different properties: Mendes--Herlihy algorithm contracts coordinate-wise ranges, while MidExtremes and MDA contract the diameter of correct values. Our algorithms contract the radius of the minimum enclosing ball of correct values. To compare MDA under the same validity notion, we derive the \MEB-validity guarantees. 

Tables~\ref{tab:sync-comparison} and~\ref{tab:async-comparison} summarize the resulting guarantees in the synchronous and asynchronous settings. Here, two resilience bounds emerge. At the convex validity thresholds $n>(d+1)t$ and $n>(d+2)t$ in the synchronous and asynchronous model, respectively, achieve smaller dimension-independent contraction factors, though for the radius contraction, at the price of satisfying $\sqrt{2}$- and $\sqrt{6}$-\MEB validity. Below those thresholds, where convex validity cannot be achieved, ball inflation extends the same algorithms to resilience $n>3t$ and $n>4t$, with weaker but constant contraction and with better validity constants than MDA achieves at worse resilience.
We next give the main technical ideas behind these results.

\begin{table}[t]
\centering
\footnotesize
\renewcommand{\arraystretch}{1.25}
\setlength{\tabcolsep}{4pt}
\begin{tabular}{|c|c|c|c|c|}
\hline
Algorithm & Resilience & Validity & Contraction & Coord. \\
\hline 
MidExtremes~\cite{fugger_et_al:LIPIcs.DISC.2018.27}
&
\shortstack{non-split\\[-1pt]{\scriptsize round model}}
&
\shortstack{convex\\[-1pt]{\scriptsize($\Rightarrow$ 1-\MEB)}}
&
\shortstack{$\displaystyle D_{r+1}\le \sqrt{7/8}D_r$\\[-1pt]{\scriptsize \cite{fugger_et_al:LIPIcs.DISC.2018.27}}}
&
\checkmark
\\
\hline
\shortstack{Adaptive \MEB Contr.\\[-1pt]{\scriptsize Alg.~\ref{alg:adaptive-contraction}, $\alpha=1$}}
&
$n>(d+1)t$
&
\shortstack{$\sqrt{2}$-\MEB\\[-1pt]{\scriptsize Thm.~\ref{thm:synch-no-inflation}, Cor.~\ref{cor:synch-alpha-one}}}
&
\shortstack{$\displaystyle R_{r+1}\le 1/\sqrt{2}R_r$\\[-1pt]{\scriptsize Lem.~\ref{lem:new-contraction}}}
&
\checkmark
\\
\hline
MDA~\cite{10.5555/3540261.3542179}
&
$n>4t$
&
\shortstack{$7$-\MEB\\[-1pt]{\scriptsize Lem.~\ref{lem:mda-validity}}}
&
\shortstack{$\displaystyle D_{r+1}\le 2/3D_r$\\[-1pt]{\scriptsize \cite{10.1007/978-3-032-11127-2_10}}}
&
\checkmark
\\
\hline
\shortstack{Inflated Adaptive \MEB Contr.\\[-1pt]{\scriptsize Inflated Alg.~\ref{alg:adaptive-contraction}, $\alpha=1$}}
&
$n>3t$
&
\shortstack{$\sqrt{6}$-\MEB\\[-1pt]{\scriptsize Thm.~\ref{thm:synch-inflation}, Cor.~\ref{cor:validity-inflated-synch}}}
&
\shortstack{$\displaystyle R_{r+1}\le \sqrt{3}/2R_r$\\[-1pt]{\scriptsize Lem.~\ref{lem:contraction-inflated}}}
&
\checkmark
\\
\hline
\end{tabular}
\caption{Comparison of synchronous and round-based multidimensional approximate-agreement algorithms. Here $R_r$ denotes the radius of the minimum enclosing ball of the correct values in round $r$, while $D_r$ denotes their diameter. For our algorithms, $\alpha=1$. The last column indicates whether the algorithm is coordinate-free. The MidExtremes is stated for the non-split network model, in which any two processes have a common incoming neighbor, rather than for the Byzantine consistent broadcast model we use. Note that convex validity implies 1-\MEB validity.}
\label{tab:sync-comparison}
\end{table}

\begin{table}[t]
\centering
\footnotesize
\renewcommand{\arraystretch}{1.25}
\setlength{\tabcolsep}{3.5pt}
\begin{tabular}{|c|c|c|c|c|}
\hline
Algorithm & Resilience & Validity & Contraction & Coord. \\
\hline
\shortstack{Mendes--Herlihy\\[-1pt]{\scriptsize \cite{mendes2015multidimensional}}}
&
$n>(d+2)t$
&
\shortstack{convex\\[-1pt]{\scriptsize($\Rightarrow$ 1-\MEB)}}
&
\shortstack{$\displaystyle \Delta^{(m)}_{r+1}\le 2^{-1/d}\Delta^{(m)}_r$\\[-1pt]{\scriptsize \cite{mendes2015multidimensional}}}
&
$\times$
\\
\hline
MidExtremes~\cite{fugger_et_al:LIPIcs.DISC.2018.27}
&
$n>(d+2)t$
&
\shortstack{convex\\[-1pt]{\scriptsize($\Rightarrow$ 1-\MEB)}}
&
\shortstack{$\displaystyle D_{r+1}\le \sqrt{7/8}D_r$\\[-1pt]{\scriptsize \cite{fugger_et_al:LIPIcs.DISC.2018.27}}}
&
\checkmark
\\
\hline
\shortstack{Adaptive \MEB Contr.\\[-1pt]{\scriptsize Alg.~\ref{alg:adaptive-contraction-asynch}, $\alpha=1$}}
&
$n>(d+2)t$
&
\shortstack{$\sqrt{6}$-\MEB\\[-1pt]{\scriptsize Thm.~\ref{thm:async-d2t}, Cor.~\ref{cor:asynch}}}
&
\shortstack{$\displaystyle R_{r+1}\le \sqrt{2/3}R_r$\\[-1pt]{\scriptsize Lem.~\ref{lem:asynch-contraction-d+2t}}}
&
\checkmark
\\
\hline
MDA~\cite{10.5555/3540261.3542179}
&
$n>7t$
&
\shortstack{$11$-\MEB\\[-1pt]{\scriptsize Lem.~\ref{lem:mda-validity}}}
&
\shortstack{$\displaystyle D_{r+1}\le 4/5D_r$\\[-1pt]{\scriptsize \cite{10.5555/3540261.3542179}}}
&
\checkmark
\\
\hline
\shortstack{Inflated Adaptive \MEB Contr.\\[-1pt]{\scriptsize Inflated Alg.~\ref{alg:adaptive-contraction-asynch}, $\alpha=1$}}
&
$n>4t$
&
\shortstack{$2\sqrt{10}$-\MEB\\[-1pt]{\scriptsize Thm.~\ref{thm:asynch-inflation}, Cor.~\ref{cor:async-4t-alpha-one}}}
&
\shortstack{$\displaystyle R_{r+1}\le \sqrt{15}/4R_r$\\[-1pt]{\scriptsize Lem.~\ref{lem:asynch-contraction-4t}}}
&
\checkmark
\\
\hline
\end{tabular}
\caption{Comparison of asynchronous multidimensional approximate-agreement algorithms. For our algorithms, $\alpha=1$. Here $R_r$ denotes the radius of the minimum enclosing ball of the correct values, $D_r$ denotes their Euclidean diameter, and $\Delta^{(m)}_r$ denotes the coordinate-wise working range in coordinate $m$. The contraction entries use different progress measures: Mendes--Herlihy contracts one coordinate range at a time, MidExtremes and MDA contract diameter, and our algorithms contract the \MEB radius. The last column indicates whether the algorithm is coordinate-free. The inflated variant applies beyond the dimension-dependent resilience limitation of convex validity, for $n>4t$.}
\label{tab:async-comparison}
\end{table}

\subsection{Technical Overview}

We consider a fully connected network of $n$ processes, where at most $t$ are Byzantine, in two communication models. In the synchronous model, computation is carried out in rounds and every process consistently broadcasts its current value, so all correct processes receive all correct values. In the asynchronous model, messages can be delayed arbitrarily and a process cannot wait for all correct values to arrive. Hence, a process obtains a local view through the Gather protocol~\cite{abraham2021reaching, canetti1993fast}, which guarantees a common-core of at least $n-t$ values shared by all correct processes. However, the Gather protocol does not reveal which of the received values form the common-core.

\paragraph{\MEB validity for approximate agreement.}
The $\MEB$-validity was introduced for vector consensus
in~\cite{cambus2026practicalvalidityconditionsbyzantinetolerant}. Extending the definition 
into the approximate agreement poses new challenges. For the vector consensus, a process decides once and the validity condition concerns that single decision. In the approximate agreement the correct values are updated in every round, however the validity condition is stated with respect to the minimum enclosing ball of the initial correct values. 
It is therefore not enough to show that each value is close to the correct \MEB of the current round. This bound must hold with respect to the initial correct \MEB, and it must hold in every round. 
We consider two cases, depending on whether the candidate balls, that is, the minimum enclosing balls of all subsets of $n-t$ received values whose intersection forms the safe area, have a common point. For $n>(d+1)t$ they do, by Helly's theorem, which allows us to apply the adaptive contraction directly. At optimal resilience $n>3t$ the balls do not necessarily intersect, and we restore the intersection by multiplicatively increasing their radii. The remainder of this overview explains how these two ingredients, adaptive contraction and ball inflation, fit together.

\paragraph{An Adaptive \MEB Contraction Algorithm.}
Agreement is reached by shrinking the region that contains all correct values. Since each process computes its next value locally from its own received values, different correct processes may output different points. The contraction argument must therefore show that these outputs are geometrically closer to each other at the end of the round than they were at the beginning. We measure this progress using the radius $R_r$ of the minimum enclosing ball of the correct values in round $r$.
The key reason contraction is possible is that the local $\MEB$-safe areas share a common point: every correct output is forced to stay close to this point, and this gives a common reference around which the new correct values can be enclosed.

Adaptive \MEB Contraction Algorithm applies the core-set idea of B\u{a}doiu and
Clarkson~\cite{badoiu2003smaller,badoiu2008optimal} to
the $\MEB$-safe area: a process repeatedly adds the point of $\safeMEB_i^{(r)}$
farthest from the current minimax center, until every point of the safe area
lies within $\alpha r_i$ of that center. As shown in Lemma~\ref{lem:new-contraction}, the parameter $\alpha$ trades
computational effort against accuracy and provides a contraction factor of
$\alpha/\sqrt{1+\alpha^{2}}$, which is $1/\sqrt2\approx0.707$ at $\alpha=1$ and improves on the $\sqrt{7/8}\approx0.935$ of MidExtremes~\cite{fugger_et_al:LIPIcs.DISC.2018.27}.
Note that the contraction bound does not depend on the dimension.

The contraction proof is geometric. The selected points lie inside the correct $\MEB$, which restricts the possible position of their minimax center $c_i$. At the same time, the stopping condition ensures that the common point $p_r$ contained in all local \MEB-safe areas is also close to $c_i$. Combining these constraints gives a quadratic inequality involving the output, the current correct center $C_r$ and the common point $p_r$. Rewriting this inequality shows that every correct output lies in a ball centered at a weighted midpoint of $C_r$ and $p_r$, denoted by $a_r$. Since this midpoint is common to all correct processes, the local inequalities imply a global bound on the next correct radius $R_{r+1}$.

\paragraph{Validity is not implied by contraction.}
Note that the radius contraction alone does not imply validity. That is because the center of the correct $\MEB$ may drift from round to round, so the contraction cannot keep the correct values close to the initial ball $B(C_0, R_0)$, centered at $C_0$ with radius $R_0$. A naive proof for bounding the drift would sum these center movements over all rounds, but this would give weaker validity constant. We instead, in Theorem~\ref{thm:synch-no-inflation}, establish the inequality for any two consecutive rounds:
$\|C_{r+1}-C_r\|+\gamma R_{r+1}\le\gamma R_r$, which states that the
balls with enlarged radius by $\gamma>1$ are nested, $B(C_{r+1},\gamma R_{r+1})\subseteq B(C_r,\gamma R_r)$. Radius
contraction then compensates for center drift within a single inequality, and induction gives $H^{(r)}\subseteq B(C_0,\gamma R_0)$ for all rounds, where $H^{(r)}$ denotes the set of correct values in round $r$. This is exactly $\gamma$-\MEB validity with respect to the initial correct values. A similar argument applies to all four settings.

\paragraph{Ball inflation.}
The previous contraction argument requires the candidate $\MEB$ balls defining the safe areas to have a common intersection. Above the Helly threshold, i.e. $n>(d+1)t$, the intersection follows from Helly's theorem. 
At optimal resilience $n>3t$, however the balls do not necessarily intersect. We therefore in
Theorem~\ref{thm:beta-wise-ball-inflation} prove a multiplicative inflation theorem for Euclidean balls. The theorem states: if every $\beta$ balls of a finite family have a common point, then inflating every radius by
$\sqrt{\beta/(\beta-1)}$ makes the entire family intersect. The factor is
dimension-free. A counting argument gives $\beta\ge 3$ for $n>3t$ in the synchronous setting, so the required inflation is at most $\sqrt{3/2}$. The inflated algorithm then uses the same contraction, but on the inflated safe area. This weakens the contraction factor from $\alpha/\sqrt{1+\alpha^2}$ to $\lambda \alpha/\sqrt{1+ \alpha^2}$. For $\alpha=1$ it gives contraction factor at most $\sqrt{3}/2$ and $\sqrt{6}$-\MEB validity. Thus, inflation gives dimension-free and optimal resilience $n>3t$, at the price of slower contraction and a slightly larger validity constant.

The proof of the inflation theorem is a dimension-free Helly-type argument in the spirit of Adiprasito, B\'ar\'any, Mustafa and Terpai~\cite{adiprasito2020theorems} but the guarantee we need is different. Their theorem is additive, whereas $c$-\MEB validity is multiplicative and requires bounds of the form $\|y-C^{*}\|\leq c R^{*}$.
We therefore need a guarantee relative to each ball's own radius. 
Let $\lambda^*$ be the smallest inflation factor for which all balls intersect, and consider the balls that are tight at a witness point. The directions from this witness point to the centers of the tight balls must balance; otherwise the witness point could be moved slightly to decrease $\lambda^{*}$. This balance may involve many balls, while the assumption only gives intersections for $\beta$ balls at a time. The key step is to average over all $\beta$-tuples of tight directions and show that there are $\beta$ of them that are already approximately balanced, with average direction of squared norm at most $1/\beta$. Solving the inequalities we obtain the bound $\lambda^{*}\le\sqrt{\beta/(\beta-1)}$.

\paragraph{Relating the safe area to the correct \MEB.}
Our contraction analysis requires the local \MEB-safe area to be contained in the correct \MEB. When $n>(d+1)t$, this holds trivially. The correct \MEB is defined by at most $d+1$ points, so there exists a locally computed candidate subset of size $n-t$ containing the same $d+1$ points. However, this argument does not hold for $n>3t$. We repair it by intersecting over all subsets of size at least $n-t$, rather than exactly size $n-t$. Since the set of correct values $H^{(r)}$ is then itself a candidate subset, the inflated correct ball $B(C_r,\lambda R_r)$ appears in the intersection, and the local safe area is contained in it without any dependence on $d$.

\paragraph{The asynchronous setting.}
In the asynchronous setting, using the Gather protocol a correct process may proceed without having received all correct
values. Then, the \MEB safe area cannot be compared to $\MEB(H^{(r)})$ directly. Instead, we first compare it to the $\MEB$ of the correct values the process received and relate that local ball to the correct one in the second step. This additional comparison weakens the quadratic inequality used in the synchronous proof. Consequently, the non-inflated asynchronous algorithm contracts by $\sqrt{2/3}$ for $\alpha=1$ at resilience $n>(d+2)t$, whereas the inflated asynchronous algorithm works at dimension-free resilience $n>4t$ with contraction factor $\sqrt{15}/4$ and $2\sqrt{10}$-\MEB validity.

\paragraph{Comparison with MDA.}
Minimum-Diameter Averaging is analyzed in terms of diameter contraction and is
known to satisfy neither box nor convex
validity~\cite{10.1007/978-3-032-11127-2_10}. Thus, the literature does not provide a directly comparable $\MEB$-validity guarantee. In Lemma~\ref{lem:mda-validity}, we derive the $\MEB$-validity bound implied by its diameter contraction. This gives $7$-\MEB validity in the synchronous model and $11$-\MEB validity in the asynchronous model. By comparison, our inflated algorithms achieve $\sqrt{6}\approx2.45$-\MEB validity synchronously and $2\sqrt{10}\approx6.32$-\MEB validity asynchronously, both under strictly better resilience thresholds.

\subsection{Roadmap}
We discuss related work in Section~\ref{sec:related-work}. In Section~\ref{sec:model}, we define the communication model, approximate agreement and $\MEB$ validity. Section~\ref{sec:synch-alg} introduces Adaptive $\MEB$ Contraction and proves its contraction guarantee. In Section~\ref{sec:synchronousMAA}, we apply this algorithm in the synchronous setting, first above the Helly's threshold and then at optimal resilience using ball inflation. Section~\ref{sec:asynchMAA} extends the approach to the asynchronous setting using the Gather protocol. Finally, Section~\ref{sec:discussion} compares our guarantees with prior algorithms and discusses limitations and future work.

\section{Related Work}\label{sec:related-work}

\paragraph{Approximate agreement.}
Approximate agreement was introduced by Dolev et al.~\cite{10.1145/5925.5931} as a relaxation of exact agreement in which correct processes do not agree on identical values, but on values that are sufficiently close to each other. 
% In one dimension, the natural validity condition requires outputs to remain in the interval spanned by the correct inputs. 
Multidimensional approximate agreement generalizes this problem to inputs in $\mathbb{R}^d$ and was studied by Mendes, Herlihy, Vaidya and Garg~\cite{mendes2015multidimensional, VectorConsensusAsynch,VectorConsensus}. The standard validity notion in this setting is convex validity, which requires every correct output to lie in the convex hull of the correct input vectors. However, convex validity leads to resilience thresholds dependent on dimension $d$. In particular, the established optimal resilience thresholds are $n>(d+1)t$ in the synchronous and $n>(d+2)t$ in the asynchronous setting~\cite{VectorConsensusAsynch,VectorConsensus, mendes2015multidimensional}. An overview of applications of convex validity in other network models can be found in~\cite{ghinea2025convex, ghinea2023multidimensional, ghinea2026network}.

\paragraph{Validity conditions and relaxations.}
 Validity conditions for Byzantine agreement and approximate agreement have been studied extensively. For binary and multi-valued exact agreement, common notions include strong validity, weak validity and correct-proposal validity~\cite{Bracha87,10.1145/800221.806707,10.1145/322186.322188,10.1145/301308.301368,10.1145/872035.872066,SIU1998157}. For approximate agreement, convex validity is the standard geometric condition, however alternative one-dimensional validity notions, such as median validity and interval validity, have also been considered~\cite{MedianValidity,intervalValidity, 10.1145/3583668.3594567}.
 In the multidimensional setting, several works consider relaxed convex validity. Xiang et al.~\cite{xiang_et_al:LIPIcs.OPODIS.2016.26} study relaxed vector consensus, including lower-dimensional projections and $(\delta,p)$-relaxed validity. Coordinate-wise relaxations lead to box validity notions, which can avoid the full convex validity barrier but are inherently tied to the choice of coordinates~\cite{intervalValidity,10.1007/978-3-032-11127-2_10}. In contrast, our work uses minimum enclosing balls, which ensure that the distance from the center of the correct minimum enclosing ball and the outputs is bounded by a constant factor of the radius. 

\paragraph{Minimum enclosing balls and \MEB validity.}
The $\MEB$-validity and the corresponding $\MEB$-safe area were introduced for vector consensus by Cambus et al.~\cite{cambus2026practicalvalidityconditionsbyzantinetolerant}. Their work proposes \MEB validity as a validity notion for vector consensus. Our work adapts \MEB validity to approximate agreement. This poses a new challenge, as the agreement is iterative, the validity must be preserved over all rounds. 
%We therefore combine $\MEB$-safe areas with a contraction analysis that controls both the radius of the current correct $\MEB$ and the drift of its center.
Concurrent work by Melnyk~\cite{melnyk2026fasterconvergencemultidimensionalapproximate} also uses minimum enclosing balls in its contraction algorithm, but applies them to the convex safe area under convex validity.
Ball validity is also motivated by practical approximate agreement, such as Minimum-Diameter Averaging (MDA) that selects a subset of small diameter and averages it~\cite{10.5555/3540261.3542179}.

\paragraph{Core-sets and algorithms for minimum enclosing balls.}
Our Adaptive \MEB Contraction algorithm is inspired by the core-set construction of B\u{a}doiu and Clarkson~\cite{badoiu2003smaller,badoiu2008optimal}. Their algorithms approximate the minimum enclosing ball of a finite point set by repeatedly adding farthest points to a small core-set. We use the same idea, however the farthest point is chosen from a continuous $\MEB$-safe area rather than from a finite input set. Additional related work on approximate minimum enclosing encompasses similar core-set constructions ~\cite{10.1145/996546.996548} and gradient-type methods~\cite{kim20201+}.

\paragraph{Communication primitives.}
Byzantine agreement protocols rely on broadcast routines which ensure the delivery of the messages. In the synchronous setting, we use consistent broadcast~\cite{cachin2001secure, srikanth1987simulating, lynch1996distributed}, closely related to Crusader Agreement~\cite{srikanth1987simulating}, which guarantees consistency among values delivered by correct processes. However, unlike reliable broadcast~\cite{Bracha87}, it does not require all correct processes to deliver a value from a faulty sender. The term consistent broadcast was later used by Cachin et al.~\cite{cachin2001secure}; see also~\cite{aguilera2021frugal, attiya2025brief} for a discussion of the broadcast routines. In the asynchronous setting, we use the Gather protocol originating with Canetti--Rabin~\cite{canetti1993fast} and later extended by Abraham et al.~\cite{optimal-resilience-asynchronous,abraham2021reaching}.

\section{Model}\label{sec:model}

\noindent\textbf{Correct and faulty processes.} We consider a distributed system consisting of $n$ processes ${1, 2, \dots n}$, of which at most $t$ are faulty. Correct processes follow the protocol, while faulty processes are corrupted throughout the entire execution. Faulty processes, also referred to as \emph{Byzantine}, know the input vectors of all other processes, as well as the agreement algorithm, and they are allowed to collaborate. Byzantine processes may send arbitrary values or omit messages, however their behavior is imposed by the communication model described below. Note that the correct processes cannot identify which processes are faulty or how many faulty processes there are; they only know the upper bound $t$ on the number of Byzantine processes.

Throughout the paper, we consider the multidimensional setting $d\geq 2$. Each process $i$ starts with an initial input value $m_i^{(0)} \in \mathbb{R}^d$. In each round $r$, process $i$ maintains a single current value $m_i^{(r)} \in \mathbb{R}^d$, which depends on the initial input value and the information received and processed during the execution of the algorithm. 
In particular, each process combines the received information into a single value, which becomes its current value for the next round.
If the current round is insignificant, we will omit $r$ from the variables and simply write $m_i$. We use $H$ to denote the set of correct, non-faulty processes and $H^{(r)}=\{m_i^{(r)} \mid i \in H\}$ to denote the current values of correct processes in round $r$. Note that $|H| \ge n-t$.

\noindent\textbf{Communication.} 
In this work, we consider both synchronous and asynchronous settings with non-authenticated channels, as it is commonly studied in the literature~\cite{10.1145/322186.322188,dolev1982efficient, 10.1145/5925.5931}.
In the synchronous setting, computation proceeds in rounds. In each round, every correct process consistently broadcasts its current value. Consistent broadcast~\cite{cachin2001secure, srikanth1987simulating, lynch1996distributed} guarantees that every correct process receives the value broadcast by each correct process. Moreover, if two processes receive a value from the same sender, then it must be the same value. Note that a Byzantine process may send its value to some correct processes and not to others. We refer to the set of values received by a process in a round as its \emph{local view}. We denote the local view of correct process $i$ in round $r$ by $P_i^{(r)}$. Thus, $H^{(r)}\subseteq P_i^{(r)}$.

In the asynchronous setting, we use a primitive~\cite{abraham2021reaching, canetti1993fast} to obtain local views. When a correct process $i$ invokes \emph{Gather} in round $r$, it obtains a local view $P_i^{(r)}$ containing values associated with their respective senders. The Gather protocol satisfies the following properties:
\begin{itemize}
    \item \textit{Common-core:} There exists a set denoted by $\gatherr^{(r)}$ of values from at least $n-t$ distinct senders, such that $\gatherr^{(r)} \subseteq P_i^{(r)}$ for every correct process $i$.
    \item \textit{Validity:} If a value associated with a correct sender $j$ is contained in $P_i^{(r)}$, then this value is $j$'s current value $m_j^{(r)}$ in round $r$.
    \item \textit{Agreement:} If two correct parties include values associated with the same sender $j$ in their local views, then these values are identical.
\end{itemize}

Note that, the set $\gatherr^{(r)}$ is common to all correct processes, although the processes do not necessarily know which values belong to $\gatherr^{(r)}$.

\noindent\textbf{Multidimensional approximate agreement.} 
In this work, we consider multidimensional approximate agreement algorithms, where the approximation is defined with respect to the Euclidean distances:
\begin{definition}[Euclidean Distance]
For any $v, v' \in \mathbb{R}^d$, the Euclidean distance between $v$ and $v'$ is
    $\|v-v'\| = \sqrt{ \sum_{i = 1}^d (v_i - v'_i)^2}$, where  $v_i$ is the projection of $v$ on coordinate $i\in [d]$. We write $\langle v, v' \rangle= \sum_{i=1}^d v_iv_i'$ for the Euclidean inner product, so that $\|v\|^2=\langle v, v \rangle.$
\end{definition}
Approximate agreement algorithms allow processes to agree on a vector, even in presence of Byzantine processes.
An algorithm that solves multidimensional approximate agreement must satisfy the following properties:
\begin{itemize}
    \item \textit{Agreement:} The output vectors of all correct processes must be within a distance $\varepsilon>0$ from each other, i.e. $\|m_i^{(r)}-m_j^{(r)}\| \leq \varepsilon$ for $i, j \in H$.
    \item \textit{Validity:} The outputs of all correct processes satisfy the specific validity condition defined with respect to the initial inputs of the correct processes.

    \item \textit{Termination:} Each correct process must terminate in finite time, i.e. decide on a final output value and stop participating in the protocol.
\end{itemize}

The standard validity notion used in multidimensional approximate agreement is convex validity. 

\begin{definition}[Convex Validity]\label{def:convexhull}
    An algorithm satisfying convex validity must output a vector inside the convex hull of correct processes.
\end{definition}

Mendes et al.~\cite{mendes2015multidimensional} provide a multidimensional approximate agreement algorithm based on the \emph{Safe Area} computation. In the algorithm, each correct process $i$ computes the intersection of all convex hulls on subsets of size $n-t$. This intersection is referred to as the Safe Area.
% Formally, the Safe Area is defined as follows.
% \begin{definition}[Safe area~\cite{mendes2015multidimensional}]
%     Consider $n$ vectors $\{m_1, \dots ,m_n\}\eqqcolon V$, where $t<n/(\max\{3,d+1\})$ of which can be Byzantine. Let $\convexhull_1,\ldots, \convexhull_{\binom{n}{n-t}}$ be the convex hulls of every subset of $V$ of size $n-t$. The \textit{safe area} is the intersection of these convex hulls: $\bigcap_{i\in \left[\binom{n}{n-t}\right]} \convexhull_i.$
% \end{definition}
Moreover, in~\cite{mendes2015multidimensional}, the authors show that satisfying convex validity requires agreeing inside the Safe Area. In order for the Safe Area to exist, the resilience must be $t<n/(\max\{3,d+1\})$. However, this requirement implies that the algorithm cannot be used in the case when $n\leq d$. Therefore, we use a relaxation of the convex validity, which allows dimension-free resilience.

\noindent\textbf{MEB validity.}
We propose minimum enclosing ball \MEB validity and its multiplicative relaxation, $c$-\MEB validity, where $c$ is a constant, first introduced for the vector consensus setting in~\cite{cambus2026practicalvalidityconditionsbyzantinetolerant}. \MEB validity relies on allowing processes to agree inside the minimum enclosing ball of correct inputs.
In the following, we formally define the \MEB and $c$-relaxed \MEB validity conditions. 

\begin{definition}[\MEB validity]
    \MEB validity condition requires that the output vector of each non-faulty process must lie inside the minimum enclosing ball of the input vectors of all non-faulty processes denoted by \emph{correct} \MEB.
\end{definition}

As authors in~\cite{cambus2026practicalvalidityconditionsbyzantinetolerant} state, the \MEB is unique, so \MEB validity is well defined. Note that the diameter between two correct vectors is not necessarily unique. Moreover, the minimum enclosing ball of correct inputs is convex and contains all of them, hence it contains their convex hull, so convex validity implies $1$-\MEB validity.
However, Cambus et al.~\cite{cambus2026practicalvalidityconditionsbyzantinetolerant} show that exact \MEB validity suffers from resilience limitations, similar to convex validity. Hence, we relax the \MEB validity by increasing the radius of the \MEB of correct processes by a factor and improve the resilience. 

\begin{definition}[$c$-\MEB validity]
    Let $H$ be the set of input vectors of correct processes, and let \MEB denote their minimum enclosing ball with center in $C^*$ and radius $R^*$. An algorithm satisfies $c$-\MEB validity if every non-faulty process outputs a vector $y$ such that $\|y - C^*\| \le c\cdot R^*$.
\end{definition}
Note that, if $c=1$ this condition is the exact \MEB validity. Throughout the work, we will name this the $c$-relaxed \MEB validity condition.
In order to satisfy \MEB validity, it is necessary to agree inside the intersection of \MEB of all possible $n-t$ subsets of vectors, as shown in~\cite{cambus2026practicalvalidityconditionsbyzantinetolerant}. This area is defined analogously to the safe area for convex validity~\cite{VectorConsensus}:

\begin{definition}[\MEB-safe area \cite{cambus2026practicalvalidityconditionsbyzantinetolerant}]
    Let $S$ be a set of vectors in $\mathbb{R}^d$, with $|S|\ge n-t$. Then, the safe area for \MEB validity, denoted $\safeMEB$, is defined as:$$\safeMEB = \bigcap\limits_{T\subseteq S, |T|=n-t} \MEB(T).$$
    In round $r$, the local $\MEB$-safe area computed by process $i$ is $\safeMEB_i^{(r)} 
=\bigcap_{\substack{T\subseteq P_i^{(r)}\\ |T|=n-t}} \MEB(T)$.
\end{definition}
For $c$-\MEB validity, we define the $c$-\safeMEB area analogously, by replacing each
minimum enclosing ball $\MEB(T) = B(C_T, R_T)$ in the intersection with the inflated ball
$B(C_T, c \cdot R_T)$.

\section{Adaptive \MEB Contraction}\label{sec:synch-alg}

In this section, we introduce the Adaptive \MEB Contraction for solving multidimensional approximate agreement with $\MEB$-validity. First, each process locally computes the \MEB-safe area $\safeMEB_i^{(r)}$ as an intersection of smallest enclosing balls on all subsets of size $n-t$.
For now, we assume that $n>(d+1)t$, so that such intersection of the balls always exists. Later we will generalize this algorithm to work on $n>3t$ and in the asynchronous setting.

The Adaptive \MEB Contraction is inspired by the core-set construction of Bădoiu and Clarkson~\cite{badoiu2003smaller, badoiu2008optimal}, which iteratively adds a point which is farthest from the center of its current minimum enclosing ball. In contrast, we apply a similar principle to the $\MEB$-safe area $\safeMEB_i^{(r)}$ and terminate once all points of $\safeMEB_i^{(r)}$ are within a predefined distance from the center.

Initially, each process $i$ computes the \MEB-safe area $\safeMEB_i^{(r)}$ and adds the two diameter points of $\safeMEB_i^{(r)}$ to the selected set $S_i$. 
Then the algorithm computes the center $c_i$ that minimizes the maximum distance to the selected points from set $S_i$.
Next, the algorithm chooses the point $q_i$, which is the farthest away from the current minimax center $c_i$. 
Then, we check how well the current center $c_i$ represents the \MEB-safe area $\safeMEB_i^{(r)}$: if the point $q_i$ is within distance $\alpha \cdot r_i$, where $r_i$ is the maximum distance from $c_i$ to the points in set $S_i$, then all points of $\safeMEB_i^{(r)}$ are within distance $\alpha \cdot r_i$. This way, the stopping criteria is satisfied, so the algorithm outputs $c_i$. Otherwise, the point $q_i$ is added to the selected set $S_i$ and the process repeats.

The selected set $S_i$ collects points of the \MEB-safe area that are most relevant for determining a center that represents the entire area. The algorithm is adaptive because it keeps adding such points only until the current center approximates all of $\safeMEB_i^{(r)}$ within the factor $\alpha$. 
The pseudocode of the Adaptive \MEB Contraction is presented in Algorithm~\ref{alg:adaptive-contraction}. Figure~\ref{fig:adaptive-meb-example} illustrates one iteration of the Adaptive \MEB Contraction algorithm.

\begin{algorithm}[H]
\caption{Adaptive \MEB Contraction (for process $i$)}
\label{alg:adaptive-contraction}
\begin{algorithmic}[1]
\Require Input value $m_i^{(r)}$, threshold $1\leq \alpha \leq \sqrt{3}$ (see Remark \ref{remark:upperbound-alpha})
\Ensure New input value $m_i^{(r+1)}$

\State consistently broadcast $m_i^{(r)}$ and receive set $P_i^{(r)}$
\State Compute the \MEB-safe area
$\safeMEB_i^{(r)}
=\bigcap_{\substack{T\subseteq P_i^{(r)}\\ |T|=n-t}}
\MEB(T)$

\State Choose two diameter points:
$a_i,b_i\in
\arg\max_{x,y\in \safeMEB_i^{(r)}}\|x-y\|$

\State Initialize the selected set
$S_i\gets{a_i,b_i}$

\Repeat
\State Compute the minimax center of the selected set:
$
c_i\in
\arg\min_{c\in\mathbb{R}^d}
\max_{s\in S_i}\|s-c\|$
\State Define the corresponding radius: 
$r_i=\max_{s\in S_i}\|s-c_i\|$
\State Find a farthest point from the current center:
$q_i\in
\arg\max_{x\in \safeMEB_i^{(r)}}\|x-c_i\|$
\If{$\|q_i-c_i\|\le \alpha r_i$}
    \State \Return $m_i^{(r+1)}\gets c_i$
\Else
    \State Add this point: $S_i\gets S_i\cup\{q_i\}$
\EndIf
\Until{all points of $\safeMEB_i^{(r)}$ within radius $\alpha r_i$}
\end{algorithmic}
\end{algorithm}

\begin{figure}[t]
\centering
\begin{tikzpicture}[scale=1.25,>=latex]
\tikzset{
    safe/.style={fill=blue!15,draw=blue!60,thick},
    cand/.style={draw=gray!70,dashed},
    selected/.style={circle,fill=black,inner sep=1.6pt},
    farthest/.style={circle,draw=blue!70,fill=blue!20,inner sep=1.8pt},
    output/.style={circle,fill=blue!75,inner sep=1.8pt},
    lab/.style={font=\scriptsize}
}
\newcommand{\safearea}{%
  (-1,0) arc[start angle=240,end angle=300,radius=2]
         arc[start angle=0,  end angle=60, radius=2]
         arc[start angle=120,end angle=180,radius=2] -- cycle}

% ---------- Panel 1: stopping test fails ----------
\begin{scope}[xshift=-4.4cm]
  \node at (0,2.35) {\textbf{Initialization}};
  \filldraw[safe] \safearea;
  \draw[cand] (0,0) circle[radius=1];
  \node[selected,label={[lab]below left:{$a$}}]  at (-1,0) {};
  \node[selected,label={[lab]below right:{$b$}}] at (1,0) {};
  \node[farthest,label={[lab]above:{$q$}}] at (0,1.7321) {};
  \node[output,label={[lab] right:{$c_0$}}] at (0,0) {};
  \draw[->,blue!70,thick] (0,0.06) -- (0,1.65);
  \node[lab,right,xshift=7pt,yshift=25pt] at (0.06,0.95) {$\|q-c_0\|=\sqrt3>r_0$};
  \node[lab,align=center] at (0,-0.95) {$S=\{a,b\}$\\ $c_0=(0,0)$, $r_0=1$};
\end{scope}

% ---------- Panel 2: add the farthest point ----------
\begin{scope}[xshift=0cm]
  \node at (0,2.35) {\textbf{Add farthest point}};
  \filldraw[safe] \safearea;
  \draw[cand] (0,0) circle[radius=1];
  \node[selected,label={[lab]below left:{$a$}}]  at (-1,0) {};
  \node[selected,label={[lab]below right:{$b$}}] at (1,0) {};
  \node[selected,label={[lab]above:{$q$}}] at (0,1.7321) {};
  \node[output,label={[lab] right:{$c_0$}}] at (0,0) {};
  \draw[->,blue!70,thick] (0.55,1.15) to[bend left=15] (0.06,1.66);
  \node[lab,right] at (0.5,1.05) {add $q$};
  \node[lab,align=center] at (0,-0.95) {$S\gets S\cup\{q\}$};
\end{scope}

% ---------- Panel 3: updated centre, test passes ----------
\begin{scope}[xshift=4.4cm]
  \node at (0,2.35) {\textbf{Updated center}};
  \filldraw[safe] \safearea;
  \draw[cand] (0,0.5774) circle[radius=1.1547];
  \node[selected,label={[lab]below left:{$a$}}]  at (-1,0) {};
  \node[selected,label={[lab]below right:{$b$}}] at (1,0) {};
  \node[selected,label={[lab]above:{$q$}}] at (0,1.7321) {};
  \node[output,label={[lab]right:{$c_1$}}] at (0,0.5774) {};
  \draw[blue!70,thick] (0,0.5774) -- (0,1.7321);
  \node[lab,right] at (0.05,1.16) {$r_1$};
  \node[lab,align=center] at (0,-0.95)
    {$c_1=(0,\tfrac{\sqrt3}{3})$, $r_1=\tfrac{2\sqrt3}{3}$\\ stopping criterion satisfied};
\end{scope}
\end{tikzpicture}
\caption{One iteration of the Adaptive $\MEB$ Contraction Algorithm with $\alpha=1$. The
blue region is the local $\MEB$-safe area computed in Line 2, here the intersection of three balls of radius $2$ centered at $(-1,0)$, $(1,0)$ and $(0,\sqrt3)$. The two diameter points $a=(-1,0)$ and $b=(1,0)$ give $c_0=(0,0)$ and $r_0=1$, and the dashed
circle is ball centered at $c_0$ with radius $r_0$. The farthest point
$q=(0,\sqrt3)$ of the safe area lies at distance $\sqrt3>r_0$ from $c_0$, so the
stopping criterion fails and $q$ is added to the selected set. Recomputing the minimax center gives $c_1=(0,\tfrac{\sqrt3}{3})$ and $r_1=\tfrac{2\sqrt3}{3}$;
the whole safe area is now contained in $B(c_1,r_1)$, so the algorithm outputs
$c_1$. }
\label{fig:adaptive-meb-example}
\end{figure}
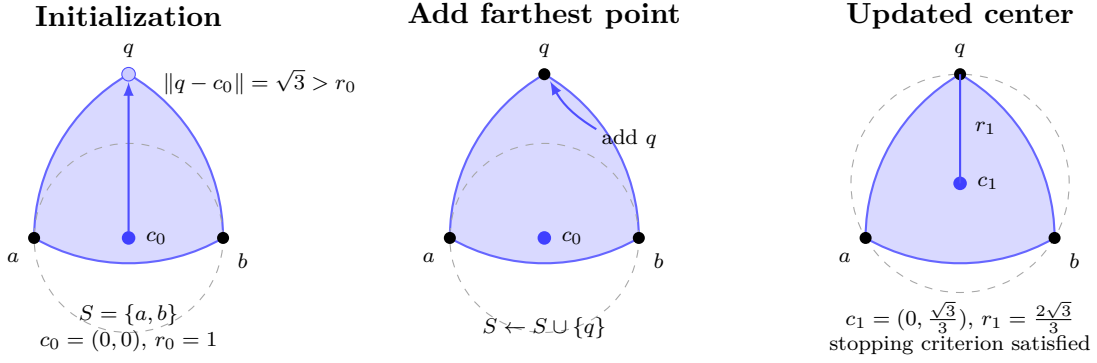

The parameter $\alpha$ controls how far away the points are allowed to lie from $c_i$. For a larger $\alpha>1$, the algorithm may stop after only a few selected points in set $S_i$. For $\alpha=1$, the algorithm stops only if the entire \MEB-safe area $\safeMEB_i^{(r)}$ is at distance $r_i$ from $c_i$, i.e. $\safeMEB_i^{(r)} \subseteq B(c_i,r_i) =\MEB(\safeMEB_i^{(r)})$. Thus, for $\alpha=1$, the output $c_i$ is exactly the center of the minimum enclosing ball of the entire \MEB-safe area. Equivalently, the $\alpha=1$ update may be implemented by computing the center of $\MEB(\safeMEB_i^{(r)})$. This case is also closely related to the approach in~\cite{melnyk2026fasterconvergencemultidimensionalapproximate}, which computes the minimum enclosing ball of the convex safe area instead.
Intermediate values of $\alpha$ interpolate between the two cases. 

The Adaptive \MEB Contraction can thus also be seen as an extension of the MidExtremes algorithm~\cite{fugger_et_al:LIPIcs.DISC.2018.27} used on the \MEB-safe area. 
MidExtremes considers only the two initial diameter points in $S_i$ and directly outputs their midpoint. In contrast, our algorithm adds further points until the current midpoint satisfies the predefined $\alpha \cdot r_i$ threshold.

Next, we analyze the contraction factor of Adaptive \MEB Contraction, which determines how quickly the correct values converge. We show that the contraction factor is $R_{r+1}\leq \frac{\alpha}{\sqrt{1+\alpha^2}}R_r$. Note that the contraction factor of our algorithm does not depend on the dimension $d$.

\begin{lemma}\label{lem:new-contraction}
    In the synchronous setting with $n>(d+1)t$, the contraction rate of Adaptive \MEB Contraction Algorithm is $R_{r+1}\leq \frac{\alpha}{\sqrt{1+\alpha^2}}R_r$, for $1 \leq \alpha \leq \sqrt{3}$. 
\end{lemma}
\begin{proof}
   Let $H^{(r)}$ denote the value of correct processes in round $r$, and $\MEB(H^{(r)})=B(C_r, R_r)$ be the smallest enclosing ball around $H^{(r)}$ with center $C_r$ and radius $R_r$. Moreover, let $R_{r+1}$ denote the radius of the new correct values after all correct processes output their $c_i$. In Adaptive \MEB Contraction Algorithm every process $i$ computes the $\MEB$-safe area $\safeMEB_i^{(r)}$. Per definition, every locally computed safe area is inside the correct $\MEB(H^{(r)})$, with center $C_r$ and radius $R_r$. Next, since $n>(d+1)t$, all locally computed $\MEB$-safe areas intersect. Indeed, consider any $d+1$ candidate balls used in the construction of the \MEB-safe areas, defined by sets $T_1,\ldots,T_{d+1}$ of size $n-t$. Each set $T_j$ omits at most $t$ processes, and therefore the $d+1$ sets omit at most $(d+1)t<n$ processes in total. Hence, there exists a process contained in all sets $T_j$. By consistent broadcast, this process contributes the same value to all corresponding sets, and this value lies in all $d+1$ candidate balls. Thus, every $d+1$ candidate balls intersect, and by Helly's theorem~\cite{danzer1963helly} the whole family of candidate balls has a nonempty intersection. Consequently, all locally computed $\MEB$-safe areas intersect, so there exists a point $p_r$ such that $p_r\in \safeMEB_i^{(r)}$ for every correct process $i$.

Each process $i$ maintains a selected set of points $S_i$ in $\safeMEB_i^{(r)}$ and computes the point that minimizes the largest distance to the selected points denoted by $c_i$. Since all selected points are within distance $r_i$ from $c_i$, then the set $S_i$ is inside the ball centered at $c_i$ with radius $r_i$, i.e. $S_i\subseteq B(c_i, r_i)$. Because $S_i \subseteq \safeMEB_i^{(r)}$, set $S_i$ is also inside the correct $\MEB(H^{(r)})$, i.e. $S_i\subseteq B(C_r, R_r)$. Since $B(c_i,r_i)$ is the smallest ball which contains $S_i$, and since $S_i$ is also contained in $B(C_r,R_r)$, we get:
\begin{align}\label{eq:contraction1}
    \|c_i-C_r\|^2 + r_i^2 \le R_r^2.
\end{align}
We refer to this property as the \MEB containment inequality.
This implies that if the radius $r_i$ is large, then the center $c_i$ must be close to the center $C_r$. If additionally $c_i$ was also far from $C_r$, then the selected set $S_i$ could not be inside $B(C_r,R_r)$.

Next, Adaptive \MEB Contraction Algorithm finds a point $q_i \in \safeMEB_i^{(r)}$ which is farthest away from the center $c_i$ and stops if $\|q_i-c_i\|\le \alpha r_i$. Then, every other point from $\safeMEB_i^{(r)}$ has a smaller distance to $c_i$ than $q_i$. This also holds for the common point $p_r$, so $\|p_r-c_i\|\le \alpha r_i$. Rewriting this gives:
\begin{align}\label{eq:contraction2}
    r_i^2\geq \frac{1}{\alpha^2}\|p_r-c_i\|^2.
\end{align}

We plug Inequality~\ref{eq:contraction2} into Inequality~\ref{eq:contraction1} and get:
\begin{align}\label{eq:synch-combined}
    \|c_i-C_r\|^2 +\frac{1}{\alpha^2}\|c_i-p_r\|^2 \leq R_r^2.
\end{align}
This inequality implies that the output $c_i$ of  Adaptive \MEB Contraction Algorithm cannot be too far from both center of the \MEB $C_r$ and common point $p_r$. 

Next, we rewrite this as squared distance from $c_i$ to the point $a_r$, which is the weighted midpoint between $C_r$ and $p_r$, i.e. $a_r=\frac{C_r+\frac{1}{\alpha^2}p_r}{1+\frac{1}{\alpha^2}}$. Thus, 
\begin{align}\label{eq:contraction-ar}
    \|c_i-C_r\|^2 +\frac{1}{\alpha^2}\|c_i-p_r\|^2 = \frac{1}{\alpha^2+1}\|C_r-p_r\|^2 + (1+\frac{1}{\alpha^2})\|c_i-a_r\|^2 \leq R_r^2.
\end{align}
From this inequality, we can conclude that $(1+\frac{1}{\alpha^2})\|c_i-a_r\|^2 \leq R_r^2$.
Hence, 
\begin{align}
    \|c_i-a_r\| \leq \frac{\alpha}{\sqrt{\alpha^2+1
    }}R_r.
\end{align}
So, every correct output $c_i$ of Adaptive \MEB Contraction Algorithm lies inside the ball $B(a_r, \frac{\alpha}{\sqrt{1+\alpha^2}}R_r)$, centered in point $a_r$ with radius $\frac{\alpha}{\sqrt{1+\alpha^2}}R_r$.
Therefore, the minimum enclosing ball of the new correct values cannot have a larger radius than this ball, i.e. $R_{r+1}\leq \frac{\alpha}{\sqrt{1+\alpha^2}}R_r$. This concludes the proof.

\end{proof}

We proved that the contraction rate of Adaptive \MEB Contraction is $\frac{\alpha}{\sqrt{1+\alpha^2}}$, where $1\leq \alpha\leq \sqrt{3}$. The parameter $\alpha$ controls the tradeoff between computational effort and contraction. Smaller values of $\alpha$ require the selected support set $S_i$ to represent the \MEB-safe area more accurately and therefore provide a stronger contraction. 
%In the extreme case $\alpha=1$, all points of $\MEB$-safe area must lie within distance $r_i$ from center $c_i$ and the contraction rate is $\frac{1}{\sqrt{2}} \approx 0.707$. 
\begin{corollary}
    For $\alpha=1$ the Adaptive \MEB Contraction Algorithm has contraction rate $\frac{1}{\sqrt{2}}\approx 0.707$.
\end{corollary}
This bound is dimension-independent and is substantially smaller than the  $\sqrt{7/8}\approx 0.935$ contraction bound shown for MidExtremes~\cite{fugger_et_al:LIPIcs.DISC.2018.27}. 
In Section~\ref{sec:discussion}, we give an example for $\alpha=1$ in which the contraction factor is exactly $\frac{1}{\sqrt{2}}$, showing that our analysis is tight.

\begin{remark}\label{remark:upperbound-alpha}
    Note that $\alpha$ does not need to be upper bounded, however we show that by taking only the initial two diameter points of the \MEB-safe area into the selected set $S_i$, Adaptive \MEB Contraction Algorithm already satisfies the stopping criterion for $\alpha=\sqrt{3}$.
\end{remark}

\begin{lemma}\label{lem:midextremes}
Let $S_i$ be initialized with the two diameter points of the $\MEB$-safe region $\safeMEB_i^{(r)}$. Then Adaptive \MEB Contraction Algorithm already satisfies the stopping criteria for $\alpha=\sqrt{3}$ after this initialization.
\end{lemma}
\begin{proof}
Let $a_i, b_i$ be the two diameter points of $\safeMEB_i^{(r)}$. Then, $c_i=\frac{a_i+b_i}{2}$ and $r_i = \frac{Diam_i}{2}$, where $Diam_i$ is the distance between $a_i$ and $b_i$. This is equivalent to the MidExtremes Algorithm~\cite{fugger_et_al:LIPIcs.DISC.2018.27} computed on the \MEB-safe area $\safeMEB_i^{(r)}$. Then, every point $x$ of $\safeMEB_i^{(r)}$ is at distance smaller than $\frac{\alpha Diam_i}{2}$ from $c_i$. Using the parallelogram law, we obtain $\|x-a_i\|^2+\|x-b_i\|^2 =2\|x-c_i\|^2+\frac{Diam_i^2}{2}$. Since $\|x-a_i\|\le Diam_i$ and $\|x-b_i\|\le Diam_i$, we get $2\|x-c_i\|^2+ \frac{Diam_i^2}{2} \leq 2Diam_i^2$. Rearranging the terms gives $\|x-c_i\| \leq \sqrt{3} \cdot r_i$, implying that all points of $\safeMEB_i^{(r)}$ are at distance $\sqrt{3} \cdot r_i$ from $c_i$. Hence, the stopping criteria in Adaptive \MEB Contraction Algorithm is directly satisfied for $\alpha=\sqrt{3}$. 
\end{proof}
This implies that taking the MidExtremes point of the \MEB-safe area provides contraction rate of $\frac{\sqrt{3}}{2} \approx 0.866$. This improves over the contraction factor $\sqrt{7/8}\approx 0.94$ proved for the classical MidExtremes algorithm in~\cite{fugger_et_al:LIPIcs.DISC.2018.27}, although the two bounds are obtained in different models.

\section{Synchronous Multidimensional Approximate Agreement}\label{sec:synchronousMAA}
In this section, we apply Adaptive \MEB Contraction in the synchronous setting. First, we consider the setting with $n>(d+1)t$ in Section~\ref{sec:resilience-n>d+1t}, in which all candidate balls for the \MEB-safe area intersect. Then, we focus on the optimal resilience case $n>3t$, where the intersection of the balls is not guaranteed. Hence, we show in Section~\ref{sec:ball-inflation} that we can inflate the candidate balls by a factor and obtain a common intersection. Then, in Section~\ref{sec:resilience3t}, we show that using the Adaptive \MEB Contraction Algorithm on inflated candidate balls solves multidimensional approximate agreement with optimal resilience $n>3t$.  

\subsection{Resilience $n>(d+1)t$}\label{sec:resilience-n>d+1t}
In the following, we show Adaptive \MEB Contraction Algorithm solves multidimensional approximate agreement when $n>(d+1)t$. In this scenario, all locally computed \MEB-safe areas intersect, as shown in Lemma~\ref{lem:new-contraction}.

\begin{theorem}
\label{thm:synch-no-inflation}
Assume the synchronous setting with $n>(d+1)t$. Adaptive \MEB Contraction Algorithm solves multidimensional approximate agreement after $O\left(
\frac{\log(R_{\max}/\varepsilon)} {\log(\sqrt{1+\alpha^2}/\alpha)} \right)$ rounds and satisfies $\sqrt{1+\alpha^2}$-\MEB validity.
\end{theorem}
\begin{proof}
We first show the convergence. As shown in Lemma~\ref{lem:new-contraction}, the contraction rate of Adaptive \MEB Contraction Algorithm is $R_{r+1} \leq \frac{\alpha}{\sqrt{1+\alpha^2}}R_r$, where $R_r$ denotes the radius of the minimum enclosing ball around correct processes in round $r$. After $r$ rounds, $R_r \leq \left( \frac{\alpha}{\sqrt{1+\alpha^2}}\right)^r R_0 \leq \left( \frac{\alpha}{\sqrt{1+\alpha^2}}\right)^r R_{\max} $, where $R_{\max}$ is the known upper bound on the initial correct \MEB radius. 
Since every correct value lies in $\MEB(H^{(r)})=B(C_r, R_r)$, the distance between any two correct values is at most $2R_r \le \varepsilon$. Hence, $2\left( \frac{\alpha}{\sqrt{1+\alpha^2}}\right)^r R_{\max} \leq \varepsilon$ and Adaptive \MEB Contraction Algorithm converges after  $\left\lceil \frac{\log (2R_{\max}/\varepsilon)}{\log (\frac{\sqrt{1+\alpha^2}}{\alpha})} \right\rceil$ rounds. Since $R_{\max}$ is known to all correct processes, they can terminate after this predetermined number of rounds. Thus, Adaptive \MEB Contraction Algorithm satisfies $\varepsilon$-agreement and termination.

Next, we show that Adaptive \MEB Contraction Algorithm satisfies $\sqrt{1+\alpha^2}$-\MEB validity. The main difficulty is that the contraction only controls the radius $R_r$. It does not prevent the center of \MEB to drift away over multiple rounds. So, for validity, we want to show that all correct values in all rounds stay close to the initial correct \MEB, i.e. $H^{(r)}\subseteq B(C_0, \gamma R_0)$ for each round $r$, where $\gamma$ is the validity factor. 

In order to bound the drift, we show that 
\begin{align}\label{eq:validity-center-drift}
\|C_{r+1}-C_r\|+\gamma R_{r+1}\le \gamma R_r. 
\end{align}
Intuitively, in every round the radius of the \MEB decreases, while the center of the \MEB of the next round might move. However, the radius contraction is strong enough to compensate for the movement.

From the proof of Lemma~\ref{lem:new-contraction}, for every correct output $c_i$, we have $\|c_i-C_r\|^2+\frac{1}{\alpha^2}\|c_i-p_r\|^2\le R_r^2$, where $p_r$ is a common point contained in every local \MEB-safe area. Additionally, consider the point $a_r=\frac{C_r+\frac{1}{\alpha^2}p_r}{1+\frac{1}{\alpha^2}}$ defined in Lemma~\ref{lem:new-contraction} and Inequality~\ref{eq:contraction-ar}. Then, 
\begin{align}
\|c_i-a_r\|^2 \leq \frac{R_r^2- \frac{1}{\alpha^2+1}\|C_r-p_r\|^2}{1+\frac{1}{\alpha^2}} =\sigma_r^2.
\end{align}
This implies that every correct output $c_i$ lies in the ball $B(a_r,\sigma_r)$, centered at $a_r$ with radius $\sigma_r$. 

Since all new correct values lie in $B(a_r,\sigma_r)$ and their minimum enclosing ball is $\MEB(H^{(r+1)})=B(C_{r+1},R_{r+1})$, the $\MEB$ containment inequality gives $\|C_{r+1}-a_r\|^2 + R_{r+1}^2 \leq \sigma_r^2$.

We now bound the validity factor $\gamma$ from Inequality~\ref{eq:validity-center-drift}. 

By triangle inequality, we get $\|C_{r+1}-C_r\|+\gamma R_{r+1}\le \|a_r-C_r\| + \|C_{r+1}-a_r\| + \gamma R_{r+1}$. Plugging in the formula for $a_r$ into the first term and using Cauchy-Schwarz and substituting the definition of $\sigma_r$ for the second and third term on the right side gives:
\begin{align}
    \|C_{r+1}-C_r\|+\gamma R_{r+1}&\le \frac{1}{\alpha^2+1}\|C_r-p_r\| + \sqrt{(1+\gamma^2)\frac{R_r^2-\frac{1}{\alpha^2+1}\|C_r-p_r\|^2}{1+\frac{1}{\alpha^2}}}\\
    &\leq \sqrt{\frac{1+\gamma^2 +\frac{1}{\alpha^2}}{1+\frac{1}{\alpha^2}}}R_r.
\end{align}
Hence, it is enough to choose a $\gamma$ such that $\sqrt{\frac{1+\gamma^2 +\frac{1}{\alpha^2}}{1+\frac{1}{\alpha^2}}} \leq \gamma$. After solving this inequality, the smallest possible choice of $\gamma$ is $\gamma=\sqrt{1+\alpha^2}$, so we obtain $\|C_{r+1}-C_r\|+\sqrt{1+\alpha^2} R_{r+1}\le \sqrt{1+\alpha^2} R_r$. This holds for every two consecutive rounds $r$ and $r+1$.

In order to show validity, we must bound the drift of the center of the \MEB from $C_0$. We now prove by induction that $\|C_{r}-C_0\|+\sqrt{1+\alpha^2} R_{r}\le \sqrt{1+\alpha^2} R_0$ for every round $r$. For $r=0$, the inequality holds. Now assume that it holds for some round $r$. Then, by the triangle inequality we get:
\begin{align}
    \|C_{r+1}-C_0\|+\sqrt{1+\alpha^2} R_{r+1}&\le \|C_r-C_0\|+\|C_{r+1}-C_r\| + \sqrt{1+\alpha^2} R_{r+1} \\
    &\leq \|C_r-C_0\| + \sqrt{1+\alpha^2} R_r \\
    &\leq \sqrt{1+\alpha^2}R_0
\end{align}
This completes the induction. Hence, for every round $r$ it holds that $$\|C_r-C_0\|+\sqrt{1+\alpha^2}R_r\le \sqrt{1+\alpha^2}R_0.$$ Since $\sqrt{1+\alpha^2}\ge 1$, this also implies $\|C_r-C_0\|+R_r \leq \sqrt{1+\alpha^2}R_0$. Now, we can use this bound to show that the correct processes also stay close to the center of the initial $\MEB(H^{(0)})$. 

Let $x\in H^{(r)}$ be any correct input in round $r$. Since $\MEB(H^{(r)})=B(C_r,R_r)$, we get $\|x-C_r\| \leq R_r$. Then by triangle inequality 
\begin{align}
    \|x-C_0\| &\leq \|x-C_r\| + \|C_r-C_0\| \\
    & \leq R_r + \|C_r-C_0\| \\
    & \leq \sqrt{1+\alpha^2}R_0.
\end{align} 
This inequality holds for every correct input and every round $r$. This implies that every correct input stays inside the ball centered at $C_0$ with radius $\sqrt{1+\alpha^2}R_0$, i.e. $H^{(r)} \subseteq B(C_0,\sqrt{1+\alpha^2}R_0)$. Thus, Adaptive \MEB Contraction Algorithm satisfies $\sqrt{1+\alpha^2}$-$\MEB$ validity.
\end{proof}

We showed that Adaptive \MEB Contraction Algorithm solves multidimensional approximate agreement in the synchronous setting with $n>(d+1)t$ while satisfying $\sqrt{1+\alpha^2}$-$\MEB$ validity. The following corollary highlights the resulting validity guarantees for particular choices of $\alpha$.
\begin{corollary}\label{cor:synch-alpha-one}
    For $\alpha = \sqrt{3}$, Adaptive \MEB Contraction Algorithm satisfies $2$-\MEB validity, whereas for $\alpha = 1$, it satisfies $\sqrt{2}\approx 1.41$-\MEB validity. In particular, lower values of $\alpha$ obtain better \MEB-validity guarantees. 
\end{corollary}

Note that in the proof of Lemma~\ref{lem:new-contraction} we can bound validity by summing over all possible center movements over all rounds. Since the radius shrinks, this bounds the total drift based on geometric series, however it provides a weaker validity constant. Hence, we proved a stronger bound by considering the drift of centers and radius contraction simultaneously.

We next focus on the optimal resilience synchronous setting $n>3t$. 
%In this regime, the local $\MEB$-safe areas no longer necessarily intersect. Hence, in order to get the intersection of the \MEB-safe areas, we inflate the candidate balls by a factor and show that the same Adaptive \MEB Contraction can be used to solve multidimensional approximate agreement. 

\subsection{Ball Inflation Analysis}\label{sec:ball-inflation}

In this section, we apply the Adaptive \MEB Contraction to solve the multidimensional approximate agreement algorithm satisfying \MEB-validity with optimal resilience $n>3t$. Note, when $n>3t$, there is no guarantee that all candidate \MEBS intersect and that the \MEB-safe area is non-empty. Hence, in order to ensure the intersection, we increase each candidate \MEB's radius by a factor of $\lambda \geq 1$ and show that all \MEBS after inflation have a common intersection point. First, we upper bound the inflation factor $\lambda \le \sqrt{\frac{\beta}{\beta -1}}$. 
Then, we show that using Adaptive \MEB Contraction Algorithm, we solve multidimensional approximate agreement with contraction factor $\frac{\lambda \alpha}{\sqrt{1+\alpha^2}}$ and satisfy  $\lambda\sqrt{\frac{1+\alpha^{2}}{1+\alpha^{2}-\lambda^{2}\alpha^{2}}}$-\MEB validity.

%In~\cite{cambus2026practicalvalidityconditionsbyzantinetolerant}, the authors showed a bound of $1+\frac{d-1}{d+1+\sqrt{2(d+1)d}} = \sqrt{\frac{2d}{d+1}} < \sqrt{2}$ on relaxation factor $c$ of the $c$-\MEB validity. Here we show that the relaxation factor in the centralized case in fact depends on $\lfloor \frac{n-1}{t}\rfloor$. In particular, for $n>3t$, the relaxation factor can be upper bounded by $\sqrt{\frac{3}{2}}$. This tighter bound helps us develop approximation algorithms that satisfy relaxed \MEB-validity. \todo{check how to mention this?}

In the following, we show that if any subset of $\beta$ balls have a non-empty intersection, then increasing each ball's radius by a factor $1 \leq \lambda \leq \sqrt{\frac{\beta}{\beta-1}}$ ensures that all balls intersect. A similar result is known for general convex sets~\cite{adiprasito2020theorems}, but only as an additive bound, whereas we need the multiplicative inflation factor relative to each candidate \MEB's own radius.

\begin{theorem} \label{thm:beta-wise-ball-inflation} Let $\mathcal{B}=\{B(C_T,R_T):T\in\mathcal{T}\}$ be a finite family of Euclidean balls in $\mathbb{R}^d$. Suppose that every subfamily of at most $\beta\ge2$ balls has a nonempty intersection. Then, after increasing the radius of every ball by a factor of $\sqrt{\frac{\beta}{\beta-1}}$ all balls have a common intersection point, i.e.
$$
\bigcap_{T\in\mathcal{T}} B\left(C_T,\sqrt{\frac{\beta}{\beta-1}}R_T\right)
\neq\emptyset.
$$
\end{theorem}
   
\begin{proof}
The proof consists of four main steps. 

\emph{(i)} Let $\lambda^*$ be the smallest inflation factor for which the inflated balls have a common point $x^*$. It suffices to show $\lambda^*\le\sqrt{\beta/(\beta-1)}$, since any inflation factor which is at least $\lambda^*$ provides a common intersection of the balls.

\emph{(ii)} We call a ball \emph{active} if $x^*$ lies on the boundary of the inflated ball and we call the unit vector pointing from $x^*$ towards the center of an inflated active ball its \emph{active direction}. We show that the active directions are balanced and cancel out. Otherwise we could move $x^*$ in one direction and reduce $\lambda^*$. 

\emph{(iii)} The cancellation in $(ii)$ may involve all active balls. However, the assumption only provides a nonempty intersection for every $\beta$ balls at a time.
By averaging over all $\beta$-tuples of active directions, we show that there are at most $\beta$ such balls that are \emph{approximately} balanced, i.e. their average direction, instead of being zero, has norm at most $\sqrt{1/\beta}$.

\emph{(iv)} Since every subset of $\beta$ balls has a common intersection point before inflation, the balls which are approximately balanced also have a common intersection point $y$. Then, $y$ lies within $R_j$ from each of their centers, while $x^*$ lies at distance $\lambda^*R_j$ from the centers. Combining the two gives the bound $\lambda^*\le\sqrt{\beta/(\beta-1)}$.

\paragraph{Step 1: The minimal inflation factor.}
Note that, if some ball has radius zero, then pairwise intersection implies that its center belongs to every ball. Hence, assume that $R_T>0$ for every $T\in\mathcal{T}$. 

Let $ \lambda^* = \min_{x\in\mathbb{R}^d} \max_{T\in\mathcal{T}} \frac{\|x-C_j\|}{R_j}$, and let $x^*$ be a minimizer. 
   Thus, $\lambda^* \geq 1$ is the smallest factor by which all balls must be enlarged to obtain a common intersection and $x^*$ is the point contained in all balls inflated by $\lambda^*$. It
suffices to prove $\lambda^*\le\sqrt{\beta/(\beta-1)}$. If $\lambda^*\le1$ the claim is immediate, so assume $\lambda^*>1$.
   Without the loss of generality, translate the coordinate system such that $x^*=0$. 

\paragraph{Step 2: The active balls are balanced.}
   
   Let $ A = \{T\in\mathcal{T}:\|C_T\|=\lambda^* R_T\}$ be the set of active balls, for which the first intersection point lies on the boundary of the inflated ball. We refer to a ball as inactive, if $\|C_T\|<\lambda^* R_T$.
   We claim that $ 0\in \operatorname{conv} \left\{ \frac{C_T}{R_T^2}:T\in A \right\}$. Suppose not. Then, by the separating hyperplane theorem~\cite{vandenberghe2004convex}, there exists a vector $h$ such that $ \left\langle h,\frac{C_T}{R_T^2} \right\rangle >0$ for every $T\in A$. Let
$f_T(x)=\|x-C_T\|^{2}/R_T^{2}$ and consider moving from $x^*=0$ to $\varepsilon h$ for $\varepsilon>0$. Expanding this term gives
$$f_T(\varepsilon h) =\frac{\|\varepsilon h-C_T\|^{2}}{R_T^{2}}=\frac{\|C_T\|^{2}}{R_T^{2}} -2\varepsilon\left\langle h,\frac{C_T}{R_T^{2}}\right\rangle+\varepsilon^{2}\frac{\|h\|^{2}}{R_T^{2}}.$$
For $T\in A$ the first term is equal to $(\lambda^*)^2$ and the second term is strictly negative, which implies that $f_T(\varepsilon h) < (\lambda^*)^2$ for all sufficiently small $\varepsilon >0$. Additionally, every inactive ball satisfies $f_T(0)<(\lambda^*)^2$.
Since $f_T$ is continuous and the inequality at $x^*$ is strict for inactive balls, $f_T(\varepsilon h)$ also stays below $(\lambda^*)^2$ for all sufficiently small $\varepsilon$.
For each ball individually, the above holds once $\varepsilon$ is small enough, but the threshold depends on the ball. Since $\mathcal{T}$ is finite, we can take the smallest of these thresholds, which is still positive. For such an $\varepsilon$, every ball satisfies $\|\varepsilon h -C_T\|<\lambda^* R_T$, so the point $\varepsilon h$ lies in all balls inflated by a factor strictly below $\lambda^*$, contradicting the minimality of $\lambda^*$.
 This shows that $x^*$ could be moved slightly in that direction and decrease the distances of all active balls.
 Therefore, there exist coefficients $a_T\ge0$, $T\in A$, such that
$\sum_{T\in A}a_T=1$ and $\sum_{T\in A}a_T\frac{C_T}{R_T^2}=0$, meaning that the directions of all active balls cancel out. 

 For $T\in A$, define the unit direction  $u_T = \frac{C_T}{\lambda^* R_T}$. Since $T$ is an active ball, $\|u_T\|=1$. Substituting $C_T=\lambda^* R_T u_T$ into the balance condition $\sum_{T\in A}a_T\frac{C_T}{R_T^2}=0$ gives 
 $$\sum_{T\in A}a_T\frac{\lambda^* u_T}{R_T}= \lambda^* \sum_{T\in A}a_T\frac{u_T}{R_T}=0.$$ Since $\lambda^*>1$, we obtain $\sum_{T\in A}\frac{a_T}{R_T}u_T=0$. After normalizing the coefficients $a_T/R_T$, we obtain weights $p_T\ge0$ satisfying $ \sum_{T\in A}p_T=1$ and $\sum_{T\in A}p_Tu_T=0$.

\paragraph{Step 3: Choosing at most $\beta$ balanced directions}
In Step 2 we showed that the directions of active balls cancel out exactly, but the cancellation involves the entire set $A$. Our assumption is however that every $\beta$ balls have a common intersection point, so we can only exploit the option of $\beta$ balls canceling out. We therefore trade exact cancellation over many directions for approximate cancellation over few, and show that some $\beta$ of the active directions already approximately cancel.

Consider all ordered $\beta$-tuples consisting of elements of the set of active balls $\boldsymbol{T}=(T_1,\ldots,T_\beta)\in A^\beta$. 
Assign to each tuple the weight $w_{\boldsymbol{T}}= \prod_{j=1}^{\beta}p_{T_j}$.
These weights are nonnegative and summing the product over all tuples gives the product of $\beta$ sums of $p_T$ over A, each of which is 1 by Step 2. Hence,
$\sum_{\boldsymbol{T}\in A^\beta}w_{\boldsymbol{T}} =\left(\sum_{T\in A}p_T\right)^\beta=1$.

For each tuple, consider the squared norm of its average direction.
We compute the weighted average of these squared norms:
$\mathcal{D}
= \sum_{\boldsymbol{T}\in A^\beta} w_{\boldsymbol{T}}
\left\| \frac{1}{\beta} \sum_{j=1}^{\beta}u_{T_j}
\right\|^2$.

Expanding the squared norm into inner products, $\mathcal{D}$ splits into $\beta$ diagonal terms, one for each position $j$, and $\beta(\beta-1)$ cross terms, one for each ordered pair of distinct positions $j$ and $k$, all divided by $\beta^2$:
$$
\mathcal{D}
= \frac{1}{\beta^2}
\left( \sum_{j=1}^{\beta}
\sum_{\boldsymbol{T}\in A^\beta}
w_{\boldsymbol{T}}\|u_{T_j}\|^2+ \sum_{\substack{j,k=1\\j\neq k}}^{\beta}
\sum_{\boldsymbol{T}\in A^\beta} w_{\boldsymbol{T}}
\langle u_{T_j},u_{T_k}\rangle \right).
$$
For each position $j$, the corresponding diagonal term satisfies
$\sum_{\boldsymbol{T}\in A^\beta}
w_{\boldsymbol{T}}\|u_{T_j}\|^2
=\sum_{T\in A}p_T\|u_T\|^2 = 1$, 
because $u_T$ is a unit vector. There are $\beta$ such terms, so together they contribute to $\frac{\beta}{\beta^2}=\frac{1}{\beta}$.
Moreover, the cross terms all vanish. Fix $j\neq k$. Since $u_{T_j}$ and $u_{T_k}$ do not involve the remaining $\beta -2$ positions, summing those out gives $ \sum_{T\in A}p_T=1$. This leaves the double sum over $j$ and $k$, which equals 0 by Step 2: $\sum_{\boldsymbol{T}\in A^\beta}
w_{\boldsymbol{T}}
\langle u_{T_j},u_{T_k}\rangle = \left\langle
\sum_{T\in A}p_T u_T, \sum_{T'\in A}p_T' u_T' \right\rangle  = 0$.

Consequently, $\mathcal{D}= \frac{\beta}{\beta^2} = \frac{1}{\beta}$. Since $\mathcal{D}$ is a weighted average with nonnegative weights summing to one, at least one tuple exists with a value not larger than the average. Fix such a tuple $(T_1,\ldots,T_\beta)\in A^\beta$, so that the squared norm of its average direction is at most $1/\beta$. Note that its indices are not necessarily distinct.

So, it remains to remove the repetitions from this tuple. Let $J\subseteq A$ be the set of distinct indices. Since the tuple has $\beta$ entries, then $|J|\le\beta$.
For each $T\in J$, let $\theta_T$ be the fraction of $\beta$ positions in which $T$ occurs. Then $\theta_T\ge0, \sum_{T\in J}\theta_T=1$, and defining $z$ as the sum of $\theta_T u_T$ over $T\in J$, i.e. $z=\sum_{T\in J}\theta_T u_T$ gives exactly the average direction of the tuple: $\|z\|^2\le\frac1\beta$. We have thus found at most $\beta$ active balls whose directions are balanced up to $1/\beta$.

\paragraph{Step 4: Bounding the inflation factor via common point}
Now we have two pieces of information about the balls in $J$. On the one hand, they are at most $\beta$ many, so by assumption they have a common point $y$ before inflation: their centers all lie within their own radii of a single location. On the other hand, they are active, so their centers lie at distance exactly $\lambda^* R_T$ from $x^*$ and by Step 3 they are spread around $x^*$ in nearly balanced directions. To compare the two, we estimate the weighted average squared distance from $y$ to the selected centers, i.e. 
$\sum_{T\in J}\mu_T\|y-C_T\|^2$.

Fix $y\in\bigcap_{T\in J}B(C_T,R_T)$ as a common point of balls from set $J$ before inflation. 
Thus, $\|y-C_T\|\le R_T$ for every $T\in J$.

Note that the balls in $J$ can have different radii. In order to normalize this, we set $P=\sum_{T\in J}\frac{\theta_T}{R_T}$ and $Q= \sum_{T\in J}\theta_T R_T$, 
and define the weights $\mu_T = \frac{\theta_T/R_T}{P}$.
The coefficients $\mu_T$ are nonnegative and sum to one. This way, reweighting by $1/R_T$ gives exactly $\sum_{T\in J}\mu_T C_T = \sum_{T\in J}\frac{\theta_T /R_T}{P}\lambda^* R_T u_T = \frac{\lambda^* z}{P}$, where $z$ is the average direction computed in Step 3.

Next we upper bound the term $\sum_{T\in J}\mu_T\|y-C_T\|^2$. Since $y\in B(C_T,R_T)$ for every $T\in J$,
\begin{align}
\sum_{T\in J}\mu_T\|y-C_T\|^2
&\le \sum_{T\in J}\mu_T R_T^2 = \frac{Q}{P}.
\label{eq:intersection-upper-bound}
\end{align}

Next, we focus on the lower bound on $\sum_{T\in J}\mu_T\|y-C_T\|^2$. Expanding the square gives:
\begin{align*}
\sum_{T\in J}\mu_T \|y-C_T\|^2
&=\left\| y-\sum_{T\in J}\mu_T C_T
\right\|^2 + \sum_{T\in J}\mu_T\|C_T\|^2 -\left\|
\sum_{T\in J}\mu_T C_T \right\|^2\\
&\ge
\sum_{T\in J}\mu_T\|C_T\|^2 -\left\|
\sum_{T\in J}\mu_T C_T \right\|^2.
\end{align*}

Since $C_T=\lambda^* R_T u_T$ for every $T\in J$, then we obtain
$\sum_{T\in J}\mu_T\|C_T\|^2
= (\lambda^*)^2\frac{Q}{P}$ for the first term and 
$\sum_{T\in J}\mu_T C_T=
\frac{\lambda^*}{P}
\sum_{T\in J}\theta_T u_T =\frac{\lambda^*}{P}z$ for the second.

Combining this with Inequality~\eqref{eq:intersection-upper-bound} gives
$$(\lambda^*)^2
\left(
\frac{Q}{P}
- \frac{\|z\|^2}{P^2}
\right) \leq \sum_{T\in J}\mu_T\|y-C_T\|^2 
\le
\frac{Q}{P}.$$
Since $Q/P >0$, this is equivalent to:
\begin{align} \label{eq:lambdacauchy}
(\lambda^*)^2
\left(
1-\frac{\|z\|^2}{PQ}
\right)
\leq 1.
\end{align}

Using Cauchy--Schwarz on $PQ$, we obtain
$$
PQ= \left( \sum_{T\in J}\frac{\theta_T}{R_T}
\right) \left( \sum_{T\in J}\theta_T R_T
\right) \geq \left( \sum_{T\in J} \theta_T \right)^2 =1.
$$
Therefore, $\frac{\|z\|^2}{PQ}
\leq \|z\|^2
\leq \frac1\beta$, and we can use this to simplify Inequality~(\ref{eq:lambdacauchy}).
Since $\beta \geq 2$, it follows that
$$
(\lambda^*)^2
\left(1-\frac1\beta\right)
\le1,
$$
and hence
$$
\lambda^*
\le
\sqrt{\frac{\beta}{\beta-1}}.
$$
Thus, the balls enlarged by the factor
$\sqrt{\frac{\beta}{\beta-1}}$ have a common intersection.

\end{proof}

\subsection{Multidimensional Approximate Agreement with Optimal Resilience $n>3t$} \label{sec:resilience3t}
%In Section~\ref{sec:resilience-n>d+1t}, we assumed that $n>(d+1)t$, which guarantees by Helly's theorem that the candidate \MEBS intersect and that the \MEB-safe area is non-empty. 
%We now improve resilience and remove the dependence on the dimension $d$, and consider the optimal resilience case $n>3t$, where the candidate \MEBS do not necessarily intersect.
We now combine the ball-inflation theorem from Section~\ref{sec:ball-inflation} with Adaptive \MEB Contraction. This gives a synchronous algorithm with dimension-free resilience $n>3t$.

Our approach is to run the Adaptive \MEB Contraction Algorithm on inflated candidate balls. In particular, each correct process $i$ in round $r$ replaces $\MEB(T)=B(C_T, R_T)$ by $B(C_T, \lambda R_T)$ and computes the inflated \MEB-safe area $\safeMEB_{i,\lambda}^{(r)}
=\bigcap_{\substack{T\subseteq P_i^{(r)}\\ |T|\geq n-t}}
B(C_T, \lambda R_T)$. In the following, we establish two things: first, the candidate \MEBS inflated by factor $\lambda$ have a non-empty intersection, so that the \MEB-safe area $\safeMEB_{i,\lambda}^{(r)}$ is non-empty as well. This follows from Theorem~\ref{thm:beta-wise-ball-inflation}, since $n>3t$ implies that every three candidate \MEBS intersect, so $\beta=3$ and inflation factor $\lambda$ is then at most $\lambda \leq \sqrt{3/2}$. Second, we need to ensure that the contraction analysis of the Adaptive \MEB Contraction Algorithm proved in Lemma~\ref{lem:new-contraction} still holds for inflated balls. In fact, this holds, however increasing the balls weakens one inequality in Lemma~\ref{lem:new-contraction}, where the \MEB-safe area is now contained in the smallest enclosing ball around correct processes inflated by factor $\lambda$. In this section, we also show that Adaptive \MEB Contraction Algorithm on inflated candidate balls satisfies  $\lambda\sqrt{\frac{1+\alpha^{2}}{1+\alpha^{2}-\lambda^{2}\alpha^{2}}}$-\MEB-validity.

\paragraph{Adaptive \MEB Contraction Modification} 
We use the same Adaptive \MEB Contraction as shown in Algorithm~\ref{alg:adaptive-contraction}, except that we replace Line 2 by the inflated \MEB-safe area:
$$ \safeMEB_{i,\lambda}^{(r)} = \bigcap_{\substack{T\subseteq P_i^{(r)}\\ |T|\geq n-t}} B(C_T,\lambda R_T)
$$
where $B(C_T,\lambda R_T)$ denotes the smallest enclosing ball around subset $T$, with its radius increased by factor $\lambda$. Two changes are made at once. First the candidate \MEBS are inflated, which provides a common intersection point if $n>3t$. Second, the candidate subsets are of size at least $n-t$, instead of exactly $n-t$. This allows us to relate to the correct \MEB inflated by $\lambda$, otherwise the containment of \MEB-safe area $\safeMEB_{i,\lambda}^{(r)}$ in the correct \MEB inflated by $\lambda$ is not guaranteed. Since every correct process receives all correct values and $|H^{(r)}|\geq n-t$, then the set $H^{(r)}$ is also a candidate subset for every correct process, so its inflated ball $ B(C_r, \lambda R_r)$ appears in the intersection. Hence, locally computed \MEB-safe areas are inside the correct \MEB, i.e. $\safeMEB_{i,\lambda}^{(r)} \subseteq B(C_r, \lambda R_r)$. 

All subsequent steps remain the same: each process initializes the diameter pair of $\safeMEB_{i,\lambda}^{(r)}$ and computes $c_i$ and $r_i$. Then, it finds the farthest point of $\safeMEB_{i,\lambda}^{(r)}$ from $c_i$ and if necessary, adds it to the selected set and repeats the process. The Adaptive \MEB Contraction Algorithm outputs $c_i$ and terminates, once the stopping criterion is satisfied. 

Note that, the inflation also restricts the admissible range of parameter $\alpha$. Namely, as we show in Lemma~\ref{lem:contraction-inflated}, the contraction requires $\alpha < \sqrt{\beta -1}$, which is $\alpha<\sqrt{2}$ for $n>3t$. 

We now apply Theorem~\ref{thm:beta-wise-ball-inflation} to the candidate \MEBS to ensure the local \MEB-safe areas have a non-empty intersection.

\begin{lemma}\label{lem:inflation-gives-intersection}
Assume the synchronous setting with $t\ge1$ and $n>3t$, and let $\beta=\left\lfloor\frac{n-1}{t}\right\rfloor$ and
$\lambda=\sqrt{\frac{\beta}{\beta-1}}$.
For every subset $T\subseteq[n]$ of size at least $n-t$ denote the smallest enclosing ball by $\MEB(T)=B(C_T,R_T)$, and let each correct process $i$ compute the inflated local safe area $\safeMEB_{i,\lambda}^{(r)}
  =\bigcap_{\substack{T\subseteq P_i^{(r)}\\ |T|\geq n-t}} B(C_T,\lambda R_T)$. Then in every round $r$ there exists a common point $p_r$ with $p_r\in \safeMEB_{i,\lambda}^{(r)}$ for every correct process $i$. Moreover, $n>3t$ implies $\beta\ge3$, and hence $\lambda\leq \sqrt{3/2} <\sqrt2$.
\end{lemma}
\begin{proof}
    Fix a round $r$ and let $\mathcal{B}_r=\bigcup_{i\in H}\bigl\{\MEB(T): T\subseteq P_i^{(r)},\ |T|\geq n-t \bigr\}$ denote the family of all candidate \MEBS computed by any correct process in round $r$.
    We first show that every subfamily of at most $\beta$ balls from $\mathcal{B}_r$ has a non-empty intersection. Let $\MEB(T_1),\dots,\MEB(T_\ell)\in\mathcal{B}_r$ with $\ell\le\beta$. Each $T_j$ has size at least $n-t$ and therefore excludes at most $t$ processes, so the $\ell$ sets together exclude at most $\ell t$ processes. Since $\ell \leq \beta = \lfloor\frac{n-1}{t}\rfloor$, we have $\ell t \leq n-1$ and hence 
    $$\Bigl|\bigcap_{j=1}^{\ell}T_j\Bigr| = n-\Bigl|\bigcup_{j=1}^{\ell}\bigl([n]\setminus T_j\bigr)\Bigr| \geq n-\ell  t \geq 1.$$
    Thus, the sets $T_1,\dots,T_\ell$ have a common process. By consistent broadcast, this process contributes the same value to every set $T_j$, and hence this value lies in $\MEB(T_j)$ for every $j$. This implies that every $\ell$ balls have a common intersection point. 

    By Theorem~\ref{thm:beta-wise-ball-inflation}, inflating the balls by factor $\lambda = \sqrt{\frac{\beta}{\beta-1}}$ gives a common point $p_r$ contained in $B(C_T,\lambda R_T)$, for every $\MEB(T)\in \mathcal{B}_r$. Since each $\safeMEB_{i,\lambda}^{(r)}$ is an intersection of a subfamily of these inflated balls, we obtain $p_r\in \safeMEB_{i,\lambda}^{(r)}$ for every correct process $i$.

    Finally, $n>3t$ gives $\beta\ge3$, so the factor $\lambda$ is bounded by $\lambda\leq \sqrt{3/2}<\sqrt{2}$.
\end{proof}

We proved that \MEBS inflated by $\lambda \leq \sqrt{\frac{\beta}{\beta -1}}$ have a common intersection point. Next, we show that Adaptive \MEB Contraction Algorithm modified for inflated candidate balls has a contraction rate $\frac{\lambda \alpha}{\sqrt{1+\alpha^2}}$ and solves multidimensional approximate agreement.

\begin{lemma}\label{lem:contraction-inflated}
In the synchronous setting with $n>3t$, the contraction rate of the Inflated Adaptive \MEB Contraction Algorithm is $R_{r+1} \leq \frac{\lambda \alpha}{\sqrt{1+\alpha^2}} R_r$ for $1\leq \alpha < \sqrt{\beta -1}$. In particular, $\alpha<\sqrt{2}$ for $n>3t$.
\end{lemma}

\begin{proof}
    Consider round $r$ and let $\MEB(H^{(r)})= B(C_r,R_r)$. By Lemma~\ref{lem:inflation-gives-intersection}, there exists a point $p_r$, such that $p_r\in \safeMEB_{i,\lambda}^{(r)}$ for every correct process $i$.

We first show $\safeMEB_{i,\lambda}^{(r)} \subseteq B(C_r,\lambda R_r)$. Due to consistent broadcast, every correct process receives all correct values, hence $H^{(r)} \subseteq P_i^{(r)}$ and $|H^{(r)}| \geq n-t$. The Inflated Adaptive \MEB Contraction considers all subsets of size at least $n-t$, then the set of correct values $H^{(r)}$ is a candidate of every correct process. Its inflated ball $B(C_r,\lambda R_r)$ appears in the intersection defining $\safeMEB_{i,\lambda}^{(r)}$, hence the \MEB-safe area is contained in $B(C_r,\lambda R_r)$. 

Now, fix a correct process $i$, and let $S_i$ be the selected set, $c_i$ the center of $S_i$ and $r_i$ the radius computed by the Inflated Adaptive \MEB Contraction Algorithm. Since $S_i\subseteq \safeMEB_{i,\lambda}^{(r)}\subseteq B(C_r,\lambda R_r)$ and $B(c_i,r_i)$ is the minimum enclosing ball around $S_i$, we get:
$$ \|c_i-C_r\|^{2}+r_i^{2} \leq \lambda^{2}R_r^{2}.$$

When the Adaptive \MEB Contraction Algorithm terminates, every point of $\safeMEB_{i,\lambda}^{(r)}$ lies within distance $\alpha r_i$ from $c_i$. This also holds for the common point $p_r$ and we get $\|p_r-c_i\|\leq \alpha r_i$, hence $r_i^{2}\geq \frac{1}{\alpha^{2}} \|c_i-p_r\|^{2}$. These are Inequalities~\eqref{eq:contraction1} and~\eqref{eq:contraction2} from Lemma~\ref{lem:new-contraction}, where $R_r$ is replaced by $\lambda R_r$. We apply the remaining steps of Lemma~\ref{lem:new-contraction} verbatim with $\lambda R_r$ instead of $R_r$ and obtain 
$$\|c_i-a_r\|\le\frac{\lambda\alpha}{\sqrt{1+\alpha^{2}}}R_r .$$
As the point $a_r$ only depends on $C_r$ and $p_r$, every correct output lies in the ball $B(a_r, \frac{\lambda \alpha}{\sqrt{1+\alpha^2}}R_r)$, centered in point $a_r$ with radius $\frac{\lambda \alpha}{\sqrt{1+\alpha^2}}R_r$.
Therefore, the minimum enclosing ball of the new correct values cannot have a larger radius than this ball, i.e. $R_{r+1}\leq \frac{\lambda \alpha}{\sqrt{1+\alpha^2}}R_r$. 

The factor $\frac{\lambda \alpha}{\sqrt{1+\alpha^2}} < 1$, when $\alpha < 1/\sqrt{\lambda^2-1} =\sqrt{\beta -1}$. For $n>3t$, this means that $1\leq \alpha <\sqrt{2}$. At $\alpha=1$ the contraction rate is $\frac{\lambda}{\sqrt2}\leq \frac{\sqrt{3}}{2} \approx 0.866$.
\end{proof}

\begin{theorem}
\label{thm:synch-inflation}
Assume a synchronous setting with $n>3t$ and let $\beta= \lfloor\frac{n-1}{t}\rfloor$, $\lambda= \sqrt{\frac{\beta}{\beta-1}}$ and $1\leq \alpha < \sqrt{\beta-1}$, then the Inflated Adaptive \MEB Contraction Algorithm solves multidimensional approximate agreement after $O\left(\frac{\log(R_{\max}/\varepsilon)}{\log(\sqrt{1+\alpha^{2}}/(\lambda\alpha))}\right)$ rounds and satisfies $\lambda\sqrt{\frac{1+\alpha^{2}}{1+\alpha^{2}-\lambda^{2}\alpha^{2}}}$-\MEB validity.
\end{theorem}

\begin{proof}
We begin by showing the convergence, similarly to the proof of Theorem~\ref{thm:synch-no-inflation}. In Lemma~\ref{lem:contraction-inflated}, we showed $R_{r+1}\leq \frac{\lambda \alpha}{\sqrt{1+\alpha^{2}}} R_r$. Over $r$ rounds, $R_r\leq \bigl(\frac{\lambda \alpha}{\sqrt{1+\alpha^{2}}}\bigr)^{r} R_0 \leq \bigl(\frac{\lambda \alpha}{\sqrt{1+\alpha^{2}}}\bigr)^{r} R_{\max}$, where $R_{\max}$ is the known upper bound on the initial correct \MEB radius. Since every
correct value lies in $\MEB(H^{(r)})=B(C_r,R_r)$, any two correct values are at distance at most $2 R_r \leq \varepsilon$. Hence, $2\bigl(\frac{\lambda \alpha}{\sqrt{1+\alpha^{2}}}\bigr)^{r} R_{\max} \leq \varepsilon$ and $\varepsilon$-agreement is reached after $\bigl \lceil \frac{\log(2R_{\max}/\varepsilon)}{\log(\frac{\sqrt{1+\alpha^2}}{\lambda \alpha})} \bigr \rceil$ rounds.

Next, we show that the Inflated Adaptive \MEB Contraction Algorithm satisfies $\lambda\sqrt{\frac{1+\alpha^{2}}{1+\alpha^{2}-\lambda^{2}\alpha^{2}}}$-\MEB validity. As in Theorem~\ref{thm:synch-no-inflation}, we bound the drift of the center and the contraction of the radius simultaneously, in the style of Inequality~\eqref{eq:validity-center-drift} for a suitable $\gamma$:
$$\|C_{r+1}-C_r\|+\gamma R_{r+1}\leq \gamma R_r.$$
In comparison to the Theorem~\ref{thm:synch-no-inflation}, only the value of $\gamma$ changes. 

From the proof of Lemma~\ref{lem:contraction-inflated}, for every correct output value $c_i$ we have $\|c_i-C_r\|^{2}+ \frac{1}{\alpha^{2}}\|c_i-p_r\|^{2} \leq  \lambda^{2}R_r^{2}$, which is the corresponding inequality from Theorem~\ref{thm:synch-no-inflation} with $\lambda R_r$ instead of $R_r$. Consequently, we define the point $a_r$ as in Theorem~\ref{thm:synch-no-inflation} and every correct output lies in $B(a_r, \sigma_r)$, where $\sigma_r^{2} 
  =\frac{\lambda^{2}R_r^{2} - \frac{1}{\alpha^{2}+1} \|C_r-p_r\|^{2}}{1+ \frac{1}{\alpha^{2}}}$. The \MEB containment inequality gives $\|C_{r+1}- a_r\|^{2} + R_{r+1}^{2}\leq \sigma_r^{2}$. 

  Using the consecutive steps and computations in the proof of Theorem~\ref{thm:synch-no-inflation} with $\lambda R_r$ instead of $R_r$, we obtain 
$$\|C_{r+1}-C_r\|+\gamma R_{r+1}\leq \lambda R_r \sqrt{\frac{1+\gamma^2 +\frac{1}{\alpha^2}}{1+\frac{1}{\alpha^2}}}.$$
Solving this inequality, the smallest possible choice of $\gamma$ is $\lambda \sqrt{\frac{(1+ \alpha^{2})}{1+\alpha^{2}-\lambda^{2} \alpha^{2}}}$. Note that $\gamma \geq 1$.

With this choice of $\gamma$, we obtain $\|C_{r+1}-C_r\|+\gamma R_{r+1}\leq \gamma R_r$ for every two consecutive rounds $r$ and $r+1$. Further, we can use the proof of Theorem~\ref{thm:synch-no-inflation} verbatim and bound the drift of the center of the \MEB from $C_0$ by induction to obtain $\|x-C_0\| \leq \lambda \sqrt{\frac{(1+ \alpha^{2})}{1+\alpha^{2}-\lambda^{2} \alpha^{2}}}$, for every correct value $x\in H^{(r)}$. Hence, $H^{(r)}\subseteq B(C_0, \lambda \sqrt{\frac{(1+ \alpha^{2})}{1+\alpha^{2}-\lambda^{2} \alpha^{2}}} R_0)$, so the Inflated 
Adaptive \MEB Contraction Algorithm satisfies $\lambda \sqrt{\frac{(1+ \alpha^{2})}{1+\alpha^{2}-\lambda^{2} \alpha^{2}}}$-\MEB validity.
\end{proof}

Unlike in Theorem~\ref{thm:synch-no-inflation}, the admissible range of $\alpha$ is now bounded by $\sqrt{2}$ and both the contraction rate and validity improve as $\alpha$ decreases. The following corollary highlights the resulting validity guarantees for particular choices of $\alpha$.

\begin{corollary}\label{cor:validity-inflated-synch}
For $\alpha=1$ the Inflated Adaptive \MEB Contraction Algorithm has contraction rate
$\frac{\lambda}{\sqrt2} \leq \frac{\sqrt{3}}{2}\approx 0.866$ and satisfies $\sqrt{\frac{2\beta}{\beta-2}}$-\MEB validity. In particular, for $n>3t$ it satisfies $\sqrt{6} \approx 2.45$-\MEB validity.
\end{corollary}

\begin{remark}
    Substituting $\lambda=1$ into the contraction factor $\frac{\lambda\alpha}{\sqrt{1+\alpha^{2}}}$ of Lemma~\ref{lem:contraction-inflated} and into the validity constant $\lambda\sqrt{\frac{1+\alpha^{2}}{1+\alpha^{2}-\lambda^{2}\alpha^{2}}}$ of Theorem~\ref{thm:synch-inflation} recovers $\frac{\alpha}{\sqrt{1+\alpha^{2}}}$ and $\sqrt{1+\alpha^{2}}$, and the admissible range $\alpha<\sqrt{\beta-1}$ becomes unbounded, as in Theorem~\ref{thm:synch-no-inflation} (see Remark~\ref{remark:upperbound-alpha}).
\end{remark}

\section{Asynchronous Approximate Agreement}\label{sec:asynchMAA}
We now turn to the asynchronous setting, where messages can be arbitrarily delayed and processes cannot wait for all correct values before proceeding with the computations. A process cannot distinguish a faulty process that never sends its value from a correct process whose message is delayed, so waiting for more than $n-t$ values risks waiting forever. Consequently, a correct process may proceed without the values of some correct processes.
Hence, each process invokes the Gather protocol~\cite{abraham2021reaching, canetti1993fast}, which returns a local view containing a common-core \gatherr of at least $n-t$ values among all correct processes. However, a process does not know which of its values belong to the common-core. This is a weakness, compared to the synchronous setting, where every correct process received all correct values, so the set $H^{(r)}$ appeared in the candidate subsets and the \MEB-safe area was compared to the correct minimum enclosing ball $\MEB(H^{(r)})$. Under the Gather protocol, we only have a guarantee that there exists a common-core of size $n-t$. However, this common-core might be missing correct values, so it is not possible to compare the candidate subsets to the $\MEB(H^{(r)})$. Instead, we compare the \MEB-safe area to the minimum enclosing ball of the correct values a process received, and then relate that ball to $\MEB(H^{(r)})$. 

First, we adapt the Adaptive \MEB Contraction Algorithm to the asynchronous setting in Section~\ref{sec:asynchalg} and analyze its contraction rate. Then, in Section~\ref{sec:async-resilience-d2t} we show that this algorithm can be used to solve multidimensional approximate agreement in the asynchronous setting with resilience $n>(d+2)t$, and in Section~\ref{sec:async-resilience-4t} with dimension-free resilience in the asynchronous setting, namely $n>4t$.

\subsection{Asynchronous Algorithm}\label{sec:asynchalg}
In this section, we adapt the Adaptive \MEB Contraction to the asynchronous setting. The algorithm is presented in Algorithm~\ref{alg:adaptive-contraction-asynch} and differs from Adaptive \MEB Contraction Algorithm only in its first two lines. 

First, a process cannot wait for all correct values to arrive, since messages can be arbitrarily delayed. Hence, each process invokes the Gather protocol, which returns a local view $P_i^{(r)}$ for process $i$ in round $r$. The Gather protocol guarantees that all correct processes receive a set $\gatherr^{(r)}$, named common-core, consisting of at least $n-t$ values. However, different correct processes may therefore hold different local views, so no process knows which of its values belong to the common-core.

Second, since the local views can differ in size, the candidate subsets are taken relative to the size of the local view. Each process $i$ intersects the minimum enclosing balls of all subsets $T\subseteq P_i^{(r)}$ with $|T|=|P_i^{(r)}|-t$. Note that $P_i^{(r)}$ contains at most $t$ faulty values, so there exists at least one subset $T$ consisting only of correct values. A natural alternative would be to fix the subset size to $|T|=n-2t$, which is the smallest size any local view admits and would make all correct processes use subsets of same size. However, the disadvantage is that a process which received more values would then discard more of them than the $t$ it must and its candidate subsets would omit correspondingly more of the common-core. Since it is the common-core that connects the local views of different processes together, this weakens the guarantee that their \MEB-safe areas intersect and requires a 
stronger resilience assumption than $n>(d+2)t$. We therefore only remove the $t$ values a process cannot trust. 

All remaining steps are unchanged: each process initializes the selected set with the diameter of $\safeMEB_i^{(r)}$, then computes the center $c_i$ and radius $r_i$, and repeatedly adds the farthest point of $\safeMEB_i^{(r)}$ until the stopping criterion is satisfied.

\begin{algorithm}[H]
\caption{Asynchronous Adaptive \MEB Contraction (for process $i$)}
\label{alg:adaptive-contraction-asynch}
\begin{algorithmic}[1]
\Require Input value $m_i^{(r)}$, threshold $1\leq \alpha < \sqrt{2}$ 
\Ensure New input value $m_i^{(r+1)}$
\State Invoke $Gather$ and obtain the local view $P_i^{(r)}$
\State Compute the \MEB-safe area
$\safeMEB_i^{(r)}
=\bigcap_{\substack{T\subseteq P_i^{(r)}\\ |T|=|P_i^{(r)}|-t}}
\MEB(T)$
\State Choose two diameter points:
$a_i,b_i\in
\arg\max_{x,y\in \safeMEB_i^{(r)}}\|x-y\|$
\State Initialize the selected set
$S_i\gets{a_i,b_i}$
\Repeat
\State Compute the minimax center of the selected set:
$
c_i\in
\arg\min_{c\in\mathbb{R}^d}
\max_{s\in S_i}\|s-c\|$
\State Define the corresponding radius: 
$r_i=\max_{s\in S_i}\|s-c_i\|$
\State Find a farthest point from the current center:
$q_i\in
\arg\max_{x\in \safeMEB_i^{(r)}}\|x-c_i\|$
\If{$\|q_i-c_i\|\le \alpha r_i$}
    \State \Return $m_i^{(r+1)}\gets c_i$
\Else
    \State Add this point: $S_i\gets S_i\cup\{q_i\}$
\EndIf
\Until{all points of $\safeMEB_i^{(r)}$ are within radius $\alpha r_i$}
\end{algorithmic}
\end{algorithm}

\subsection{Resilience $n>(d+2)t$}\label{sec:async-resilience-d2t}
In the following, we show that the contraction factor of Asynchronous Adaptive \MEB Contraction Algorithm is $\alpha\sqrt{\frac{2}{2+\alpha^2}}$. Note that the contraction factor does not depend on dimension $d$, only on the parameter $\alpha$, which takes values from 1 to $\sqrt{2}$.

\begin{lemma}\label{lem:asynch-contraction-d+2t}
In the asynchronous setting with $n>(d+2)t$ and $1\le \alpha< \sqrt2$, the contraction rate of the Asynchronous Adaptive \MEB Contraction Algorithm is $R_{r+1} \leq \alpha\sqrt{\frac{2}{2+\alpha^2}} R_r$.
\end{lemma}
\begin{proof}
    Let $H^{(r)}$ denote the values of correct processes in round $r$, and let $\MEB(H^{(r)})=B(C_r,R_r)$
be the minimum enclosing ball around the correct values. Moreover, let $R_{r+1}$ denote the radius of the minimum enclosing ball around the new correct values after all correct processes output their $c_i$.

We first show that all locally computed $\MEB$-safe areas have a common point. Consider the family of all candidate balls computed by correct processes in round $r$:
$$\mathcal{B}_r =\bigcup_{i\in H} \left\{ \MEB(T): T \subseteq P_i^{(r)}, |T|=|P_i^{(r)}|-t \right\}.$$ We show that every $d+1$ balls in $\mathcal{B}_r$ have a nonempty intersection. Take any $d+1$ sets $T_1,\ldots, T_{d+1}$ defining such balls.
For every $j$ there is a correct process $\pi_j$ with $T_j \subseteq P_{\pi_j}^{(r)}$ and $|T_j| = |P_{\pi_j}^{(r)}|-t$. By the Gather protocol, there exists a common-core $\gatherr^{(r)}$ of values from at least $n-t$ senders, with $\gatherr^{(r)} \subseteq P_{\pi_j}^{(r)}$ for every $j$. Since $T_j$ is obtained by removing exactly $t$ values from $P_{\pi_j}^{(r)}$, it can omit at most $t$ values from $\gatherr^{(r)}$, i.e. $|\gatherr^{(r)} \setminus T_j| \leq t$. From the union we obtain: 
$$\Bigl|\gatherr^{(r)}\cap \bigcap_{j=1}^{d+1}T_j\Bigr| \geq |\gatherr^{(r)}|-\sum_{j=1}^{d+1}|\gatherr^{(r)}\setminus T_j|\geq (n-t)-(d+1)t = n-(d+2)t > 0.$$
Hence, there is a sender, whose value is contained in all $T_1,\ldots, T_{d+1}$ and this value lies in all minimum enclosing balls defined on these sets. Every $d+1$ balls therefore intersect and by Helly's theorem~\cite{danzer1963helly} the whole family of candidate balls has a nonempty intersection. Consequently, all locally computed $\MEB$-safe areas intersect, so there exists a point $p_r$ such that $p_r\in \safeMEB_i^{(r)}$ for every correct process $i$.

Unlike in the synchronous case, we cannot directly compare the \MEB-safe area $\safeMEB_i^{(r)}$ to the correct minimum enclosing ball $\MEB(H^{(r)})$ because process $i$ may not receive all correct values. Instead, we compare $\safeMEB_i^{(r)}$ to the \MEB of the correct values received by process $i$. Let $H_i^{(r)}$ be the set of correct values contained in $P_i^{(r)}$ and write $\MEB(H_i^{(r)}) = B(\widehat{C}_{i,r}, \widehat{R}_{i,r})$ centered in $\widehat{C}_{i,r}$ with radius $\widehat{R}_{i,r}$. We now show that $\safeMEB_i^{(r)}\subseteq \MEB(H_i^{(r)})$. At most $t$ values in $P_i^{(r)}$ are Byzantine, so $|H_i^{(r)}| \geq |P_i^{(r)}| -t$. Next, let $D_i$ be the set of points which defines the $\MEB(H_i^{(r)})$. Since at most $d+1$ points can define a ball, then $|D_i| \leq d+1$. From $|P_i^{(r)}| \geq n-t$, we have $|P_i^{(r)}| - t \geq n-2t$ and together with $n>(d+2)t$, this gives $|P_i^{(r)}|-t \geq d+1$. Hence, we can extend $D_i$ to a subset consisting only of correct values $U_i^{(r)} \subseteq H_i^{(r)}$ of size $|P_i^{(r)}|-t$. From $D_i\subseteq U_i^{(r)} \subseteq H_i^{(r)}$, we get $\MEB(U_i^{(r)}) = \MEB(H_i^{(r)})$ and since $U_i^{(r)}$ is one of the candidate subsets computed by process $i$, its minimum enclosing ball appears in the intersection defining $\safeMEB_i^{(r)}$, so $\safeMEB_i^{(r)}\subseteq \MEB(H_i^{(r)})$. 

In Asynchronous Adaptive \MEB Contraction Algorithm, each process $i$ maintains a selected set $S_i$ in $\safeMEB_i^{(r)}$ and computes the point that minimizes the largest distance to the selected points denoted by $c_i$. Since all selected points are within distance $r_i$ from $c_i$, we have $S_i \subseteq B(c_i, r_i)$, and because $S_i \subseteq \safeMEB_i^{(r)} \subseteq \MEB(H_i^{(r)})$, the set $S_i$ is also contained in $\MEB(H_i^{(r)}) = B(\widehat{C}_{i,r}, \widehat{R}_{i,r})$. Since $B(c_i, r_i)$ is the smallest ball which contains $S_i$:
\begin{align}\label{eq:asynch-local}
    \|c_i-\widehat{C}_{i,r}\|^{2}+r_i^{2} \leq \widehat{R}_{i,r}^{2}.
\end{align}
Applying the same argument $H_i^{(r)}\subseteq H^{(r)} \subseteq \MEB(H^{(r)})$ relates the locally computed correct ball $\MEB(H_i^{(r)})$ to the correct $\MEB(H^{(r)})=B(C_r, R_r)$:
\begin{align}\label{eq:asynch-local-to-global}
    \|\widehat{C}_{i,r}-C_r\|^{2}+\widehat{R}_{i,r}^{2} \leq R_r^{2}.
\end{align}
Next, Asynchronous Adaptive \MEB Contraction Algorithm stops only when every point of $\safeMEB_i^{(r)}$ is within distance $\alpha r_i$ from $c_i$. This also holds for the common point $p_r$, so $\| p_r -c_i\|\leq \alpha r_i$ and 
\begin{align} \label{eq:asynch-stopping}
     r_i^{2} \geq \frac{1}{\alpha^{2}}\|c_i-p_r\|^{2}.
\end{align}
Similarly to the proof of Lemma~\ref{lem:new-contraction}, plugging in Inequality~\ref{eq:asynch-stopping} into Inequality~\ref{eq:asynch-local} and eliminating $\widehat{R}_{i,r}^{2}$ with Inequality~\ref{eq:asynch-local-to-global}, we obtain 
\begin{align}\label{eq:three-points}
\|c_i-\widehat{C}_{i,r}\|^{2}+\|\widehat{C}_{i,r}-C_r\|^{2}+\frac{1}{\alpha^{2}} \|c_i-p_r\|^{2} \leq R_r^{2}.
\end{align}
Using the parallelogram law on the first two terms of Inequality~\ref{eq:three-points}, we get 
\begin{align}
    \frac{1}{2}\|c_i-C_r\|^{2}+\frac{1}{\alpha^{2}}\|c_i-p_r\|^{2} \leq R_r^{2}  
\end{align}
which is analogous to the Inequality~\ref{eq:synch-combined} in Lemma~\ref{lem:new-contraction}. Note that the coefficient of $\|c_i-C_r\|^{2}$ dropped from 1 to $1/2$ which is the result of comparing the $\safeMEB_i^{(r)}$ to the local ball instead of directly to $\MEB(H^{(r)})$.

Along the lines of Lemma~\ref{lem:new-contraction}, we define $a_r$ as the weighted midpoint between $C_r$ and $p_r$, i.e. $a_r=(\frac{1}{2}C_r + \frac{1}{\alpha^2}p_r)/(\frac{1}{2} + \frac{1}{\alpha^2})$. Thus, 
\begin{align}\label{eq:asynch-infl-drop-pr}
    \frac{1}{2} \|c_i-C_r\|^2
+\frac{1}{\alpha^2}\|c_i-p_r\|^2 = \left(\frac{1}{2}+\frac{1}{\alpha^2}\right)\|c_i-a_r\|^2
+ \frac{1}{\alpha^2+2}\|C_r-p_r\|^2 \leq R_r^2.
\end{align}
Hence, $\|c_i-a_r\|\leq  \alpha\sqrt{\frac{2}{2+\alpha^{2}}} R_r$. So every correct output $c_i$ lies inside the ball $B(a_r, \alpha\sqrt{\frac{2}{2+\alpha^{2}}} R_r)$, centered in $a_r$ with radius $\alpha\sqrt{\frac{2}{2+\alpha^{2}}} R_r$. Therefore, the minimum enclosing ball of the new correct values cannot have a larger radius than this ball, i.e. $R_{r+1}\leq \alpha\sqrt{\frac{2}{2+\alpha^{2}}} R_r$. This factor is smaller than one exactly when $\alpha < \sqrt{2}$. 
\end{proof}
We showed that the contraction rate of Asynchronous Adaptive \MEB Contraction Algorithm is $\alpha\sqrt{\frac{2}{2+\alpha^{2}}}$ for $1\leq \alpha < \sqrt{2}$. Note that the contraction is dimension-independent. Next, we show that Asynchronous Adaptive \MEB Contraction Algorithm solves multidimensional approximate agreement under \MEB-validity.

\begin{theorem}
\label{thm:async-d2t}
Assume $n>(d+2)t$ and $1\le \alpha<\sqrt2$. Then Asynchronous Adaptive \MEB Contraction Algorithm solves multidimensional approximate agreement after
$O\left(\frac{\log(R_{\max}/\varepsilon)}
{\log\left(\frac{\sqrt{2+\alpha^2}}{\alpha \sqrt{2}}\right)}
\right)$ rounds and satisfies $\sqrt{\frac{2(\alpha^2+2)}{2-\alpha^2}}$-\MEB validity.

\end{theorem}
\begin{proof}
We first show the convergence. In Lemma~\ref{lem:asynch-contraction-d+2t}, we showed that the radius of the minimum enclosing ball shrinks $R_{r+1}\leq \alpha\sqrt{\frac{2}{2+\alpha^{2}}} R_r$, with $1\leq \alpha < \sqrt{2}$. The
argument of Theorem~\ref{thm:synch-no-inflation} applies verbatim with contraction factor $\alpha\sqrt{\frac{2}{2+\alpha^{2}}}$ and gives convergence after $\left\lceil \frac{\log(2R_{\max}/\varepsilon)}{\log\left(\frac{\sqrt{2+ \alpha^2}}{\alpha \sqrt{2}}\right)} \right\rceil$ rounds. This proves convergence and termination. 

Next, we show that Asynchronous Adaptive \MEB Contraction Algorithm satisfies $\sqrt{\frac{2(\alpha^2+2)}{2-\alpha^2}}$-\MEB validity.
As in Theorem~\ref{thm:synch-no-inflation}, we bound the drift
of the center and the contraction of the radius simultaneously and establish
\begin{equation}\label{eq:asynch-center-drift}
  \|C_{r+1}-C_r\|+\gamma R_{r+1}\leq \gamma R_r
\end{equation}
for a suitable $\gamma \geq 1$. Once this holds, the induction of Theorem~\ref{thm:synch-no-inflation} gives $H^{(r)}\subseteq B(C_0,\gamma R_0)$ for every round $r$. 

From Lemma~\ref{lem:asynch-contraction-d+2t} we have $ \frac{1}{2}\|c_i-C_r\|^{2}+\frac{1}{\alpha^{2}}\|c_i-p_r\|^{2} \leq R_r^{2} $ for every output $c_i$ and with the point $a_r$ as defined there, every correct output lies inside the ball $B(a_r, \sigma_r)$, where $ \sigma_r^{2}
  =\frac{2\alpha^{2}}{\alpha^{2}+2}
   \left(R_r^{2}-\frac{\|C_r-p_r\|^{2}}{\alpha^{2}+2}\right)$. 
Hence, all new correct inputs lie inside the ball $H^{(r+1)}\subseteq B(a_r,\sigma_r)$ and the \MEB containment inequality gives $\|C_{r+1}-a_r\|^{2}+ R_{r+1}^{2}\leq \sigma_r^{2}$.
We now bound the validity factor $\gamma$. By the triangle inequality, $\|C_{r+1}-C_r\|+\gamma R_{r+1} \leq \|a_r-C_r\|+\|C_{r+1}-a_r\|+\gamma R_{r+1}$. Using the definition of $a_r$ for the first term on the right side, using Cauchy-Schwarz and definition of $\sigma_r$ for the second and third term gives:
\begin{align}\label{eq:asynch-validity}
    \|C_{r+1}-C_r\|+\gamma R_{r+1} \leq 
    \frac{2\|p_r-C_r\|}{\alpha^2+ 2} + \sqrt{1+\gamma^2} \sqrt{\frac{2\alpha^2}{\alpha^2+2}R_r^2 -
    \frac{2\alpha^2}{(\alpha^2+2)^2}\|p_r-C_r\|^2}.
\end{align}
Hence, it is enough to choose a $\gamma$ such that $\sqrt{\frac{2}{\alpha^2}+1+\gamma^2} \sqrt{\frac{2\alpha^2}{\alpha^2 +2}} \leq \gamma$. Solving this inequality gives the smallest possible choice of $\gamma$, that is $\gamma = \sqrt{\frac{2(\alpha^2+2)}{2 -\alpha^2}}$. Plugging $\gamma$ into Inequality~\eqref{eq:asynch-center-drift} and performing the same analysis as in the proof of Theorem~\ref{thm:synch-no-inflation} we get $H^{(r)}\subseteq B(C_0,\gamma R_0)$. Thus, Asynchronous Adaptive \MEB Contraction Algorithm satisfies $\sqrt{\frac{2(\alpha^2+2)}{2 -\alpha^2}}$-\MEB validity.

\end{proof}

\begin{corollary}\label{cor:asynch}
    For $\alpha=1$, the contraction factor is $\sqrt{\frac{2}{3}}$ and the algorithm satisfies $\sqrt{6}$-\MEB validity.
\end{corollary}

We next consider the resilience case $n>4t$, which removes the dependence on the dimension $d$. As in the synchronous setting, we inflate the candidate balls by a factor $\lambda>1$ to obtain a common intersection point. The counting argument above shows that under Gather protocol every $l$ candidate balls intersect when $n>(l+1)t$, so for $n>4t$ every three balls intersect. Therefore, we apply Theorem~\ref{thm:beta-wise-ball-inflation} with $\beta=3$. We also enlarge the family of candidate subsets, as in Section~\ref{sec:resilience3t}, since the containment of the locally computed safe area \safeMEB in the correct \MEB cannot be obtained by extending a defining set. 

\subsection{Multidimensional Approximate Agreement with Resilience $n>4t$}\label{sec:async-resilience-4t}
In order to get the intersection of the candidate balls, we must first modify Asynchronous Adaptive \MEB Contraction Algorithm and inflate the candidate balls' radius by a factor $\lambda>1$. The modification is similar to the modification done for Adaptive \MEB Contraction Algorithm in the synchronous setting.

\paragraph{Inflated Adaptive \MEB Contraction}
We use the same Asynchronous Adaptive \MEB Contraction as shown in Algorithm~\ref{alg:adaptive-contraction-asynch}, except that we replace Line~2 by the inflated $\MEB$-safe area
$$\safeMEB_{i,\lambda}^{(r)} = \bigcap_{\substack{T\subseteq P_i^{(r)}\\ |T|\geq |P_i^{(r)}|-t}}
B(C_T,\lambda R_T)$$
where $B(C_T,\lambda R_T)$ denotes the smallest enclosing ball around subset $T$ with its radius increased by factor $\lambda$. 
Two changes are made at once. First, the candidate $\MEB$ balls are inflated, which provides a common intersection point when $n>4t$. Second, the candidate subsets are of size at least $|P_i^{(r)}|-t$, instead of exactly $|P_i^{(r)}|-t$. This allows us to relate the inflated $\MEB$-safe area to the $\MEB$ of the correct values received by process $i$. Indeed, if $H_i^{(r)}$ denotes the set of correct values contained in $P_i^{(r)}$, then $|H_i^{(r)}|\ge |P_i^{(r)}|-t$. Hence, $H_i^{(r)}$ is itself a candidate subset and its inflated ball appears in the intersection defining the safe area $\safeMEB_{i,\lambda}^{(r)}$.

All subsequent steps remain the same: each process initializes the selected set with the diameter pair of $\safeMEB_{i,\lambda}^{(r)}$ and computes $c_i$ and $r_i$. Then it repeatedly adds the farthest point of $\safeMEB_{i,\lambda}^{(r)}$ until the stopping criterion is satisfied.

Note that the inflation restricts the admissible range of $\alpha$. As we show below, the contraction requires $\lambda\alpha
\sqrt{\frac{1+\lambda^2}{1+\lambda^2+\alpha^2\lambda^2}} <1$, which is $\alpha<\frac{\sqrt{1+\lambda^2}}{\lambda^2}$. For the guaranteed inflation factor $\lambda=\sqrt{3/2}$, the bound on $\alpha$ becomes $1\le \alpha<\frac{\sqrt{10}}{3}$.

Next, we show that increasing the radius of candidate balls provides a common intersection point. 

\begin{lemma}\label{lem:asynch-inflation-common-point}
    Assume the asynchronous setting using the Gather protocol and $n>4t$. Let $\beta= \left\lfloor\frac{n-t-1}{t}\right\rfloor$ and
    $\lambda= \sqrt{\frac{\beta}{\beta-1}}$. Then, in every round $r$, there exists a common point $p_r$ such that $p_r\in  \safeMEB_{i,\lambda}^{(r)} $ for every correct process $i$. In particular, $n>4t$ implies $\beta\geq 3$, and hence $\lambda\leq \sqrt{3/2}$.
\end{lemma}
\begin{proof} 
Fix a round $r$ and let $\mathcal{B}_r=\bigcup_{i\in H}\bigl\{\MEB(T): T\subseteq P_i^{(r)},\ |T|\geq |P_i^{(r)}|-t \bigr\}$ denote the family of all candidate \MEBS computed by any correct process in round $r$.
We first show that every subfamily of at most $\beta$ balls from $\mathcal{B}_r$ has a non-empty intersection. Take any $\ell \leq \beta$ candidate sets $T_1,\dots,T_\ell$, where for each $j\in\{1,\dots, \ell\}$ there is a correct process $\pi_j$ such that $T_j\subseteq P_{\pi_j}^{(r)}$ and $|T_j|\geq |P_{\pi_j}^{(r)}| -t$. By the Gather protocol, there exists a common-core $\gatherr^{(r)}$ of values from at least $n-t$ senders, with $\gatherr^{(r)}\subseteq P_{\pi_j}^{(r)}$ for every $j$. Since $T_j$ excludes at most $t$ values from $P_{\pi_j}^{(r)}$, it also omits at most $t$ values from $\gatherr^{(r)}$, that is $|\gatherr^{(r)} \setminus T_j|\leq t$. Therefore, 
$$\left| \gatherr^{(r)} \cap \bigcap_{j=1}^{\ell}T_j
\right| \geq |\gatherr^{(r)}|- \sum_{j=1}^{\ell} |\gatherr^{(r)}\setminus T_j| \geq (n-t)-\ell t \geq 1.$$
Thus, there is a process whose value is contained in all sets $T_1,\dots,T_\ell$ and this value lies in each of the balls $\MEB(T_1),\dots,\MEB(T_\ell)$. Therefore, every $\beta$ balls from $\mathcal{B}_r$ intersect and for $n>4t$ we can apply Theorem~\ref{thm:beta-wise-ball-inflation} with $\beta=3$. Inflating every ball by a factor $\lambda= \sqrt{3/2}$ gives a common intersection point $p_r$ with $p_r\in B(C_T,\lambda R_T)$ for every $\MEB(T) \in \mathcal{B}_r$. Since each $\safeMEB_{i,\lambda}^{(r)}$ is an intersection of a subfamily of these inflated balls, we get $p_r\in \safeMEB_{i,\lambda}^{(r)}$ for every correct process $i$.

\end{proof}

We proved that the \MEBS inflated by $\lambda \leq \sqrt{\frac{\beta}{\beta-1}}$ have a common intersection point. Next, we show that the Inflated Asynchronous Adaptive \MEB Contraction Algorithm for inflated candidate balls provides a contraction rate $\lambda\alpha \sqrt{\frac{1+ \lambda^2}{1+\lambda^2+\alpha^2\lambda^2}}$ and solves multidimensional approximate agreement.

\begin{lemma}\label{lem:asynch-contraction-4t}
In the asynchronous setting with $n>4t$, the contraction rate of the Inflated Asynchronous Adaptive \MEB Contraction Algorithm is $R_{r+1} \leq \lambda\alpha \sqrt{\frac{1+ \lambda^2}{1+\lambda^2+\alpha^2\lambda^2}} R_r$ for $1 \leq \alpha < \frac{\sqrt{1+\lambda^2}}{\lambda^2}$. In particular, $1\le \alpha<\frac{\sqrt{10}}{3}$ for $n>4t$.
\end{lemma}
\begin{proof}
    Let $\MEB(H^{(r)})=B(C_r,R_r)$ be the minimum enclosing ball around correct values in round $r$. By Lemma~\ref{lem:asynch-inflation-common-point} there exists a point $p_r$ with $p_r\in \safeMEB_{i,\lambda}^{(r)}$ for every correct process $i$. 

    We first relate the inflated safe area to the $\MEB$ of the correct values received by process $i$. Fix a correct process $i$ and let $H_i^{(r)}$ be the set of correct values contained in the local view $P_i^{(r)}$ and write $\MEB(H_i^{(r)})=B(\widehat{C}_{i,r},\widehat{R}_{i,r})$. At most $t$ values in the local view are Byzantine, so $|H_i^{(r)}|\geq |P_i^{(r)}| -t$. Next, the Inflated Asynchronous Adaptive \MEB Contraction Algorithm considers all subsets of size at least $|P_i^{(r)}|-t$, and therefore $H_i^{(r)}$ is one of the candidate subsets. Hence, its inflated ball appears in the intersection defining $\safeMEB_{i,\lambda}^{(r)}$ and we get $\safeMEB_{i,\lambda}^{(r)} \subseteq B(\widehat{C}_{i,r},\lambda \widehat{R}_{i,r})$. Using the Inequality~\eqref{eq:asynch-local-to-global} from Lemma~\ref{lem:asynch-contraction-d+2t}, the local correct $\MEB(H_i^{(r)})$ relates to the correct $\MEB(H^{(r)})$ with 
    $$\|\widehat{C}_{i,r} -C_r\|^2 + \widehat{R}_{i,r}^2 \leq R_r^2.$$
The remaining steps follow from Lemma~\ref{lem:asynch-contraction-d+2t}. The selected set from the Inflated Asynchronous Adaptive \MEB Contraction Algorithm satisfies $S_i\subseteq \safeMEB_{i,\lambda}^{(r)} \subseteq B(\widehat{C}_{i,r},\lambda \widehat{R}_{i,r})$ and since $B(c_i, r_i)$ is its minimum enclosing ball, gives $$\|c_i-\widehat C_{i,r}\|^{2} + r_i^{2} \leq \lambda^{2} \widehat{R}_{i,r}^{2}.$$ 
The stopping criterion also holds for the common point $p_r$, hence $r_i^{2}\geq \frac{1}{\alpha^{2}} \|c_i-p_r\|^{2}$. Combining these inequalities, we get
\begin{align}
    \|c_i-\widehat C_{i,r}\|^{2} +\lambda^{2} \|\widehat C_{i,r}-C_r\|^{2} +\frac{1}{\alpha^{2}}\|c_i-p_r\|^{2} \leq \lambda^{2} R_r^{2}.
\end{align}
In contrast to Inequality~\eqref{eq:three-points}, the middle term is multiplied by the inflation factor. Thus, we perform the same step and define $a_r = ({\frac{\lambda^2}{1+ \lambda^2} C_r+ \frac{1}{\alpha^2}p_r})/(\frac{\lambda^2}{1+  \lambda^2}+\frac{1}{\alpha^2})$. Completing the square in Inequality~\eqref{eq:asynch-infl-drop-pr} and dropping the non-negative term containing $\|C_r-p_r\|^2$ gives:
$$
\|c_i-a_r\| \leq \lambda\alpha \sqrt{\frac{1+ \lambda^2}{1+\lambda^2 + \alpha^2 \lambda^2}}R_r.$$
So, every correct output $c_i$ lies inside the ball $B(a_r, \lambda\alpha \sqrt{\frac{1+ \lambda^2}{1+\lambda^2 + \alpha^2 \lambda^2}}R_r)$, centered in $a_r$ with radius $\lambda\alpha \sqrt{\frac{1+ \lambda^2}{1+\lambda^2 + \alpha^2 \lambda^2}}R_r$. Therefore, the minimum enclosing ball of the new correct values cannot have a larger radius, i.e. $R_{r+1}\leq \lambda\alpha \sqrt{\frac{1+ \lambda^2}{1+\lambda^2 + \alpha^2 \lambda^2}}R_r$. This factor is smaller than one exactly when $\alpha<\frac{\sqrt{1 +\lambda^2}}{\lambda^2}$.
\end{proof}

We showed that the contraction rate of the Inflated Asynchronous Adaptive \MEB Contraction Algorithm is $\lambda\alpha \sqrt{\frac{1+ \lambda^2}{1+\lambda^2 + \alpha^2 \lambda^2}}$ for $1\leq \alpha < \frac{\sqrt{1 +\lambda^2}}{\lambda^2}$. Next, we show that the Inflated Asynchronous Adaptive \MEB Contraction Algorithm solves multidimensional approximate agreement under \MEB-validity.

\begin{theorem}\label{thm:asynch-inflation}
    Assume asynchronous setting with $n>4t$, and let $\beta= \left\lfloor\frac{n-t-1}{t}\right\rfloor$, $\lambda= \sqrt{\frac{\beta}{\beta-1}}$ and $1\leq \alpha < \frac{\sqrt{1 +\lambda^2}}{\lambda^2}$. Then, the Inflated Asynchronous Adaptive \MEB Contraction Algorithm solves multidimensional approximate agreement after $O\left( \frac{\log(R_{\max}/ \varepsilon)}{\log\left(\frac{\sqrt{1+ \lambda^2 + \alpha^2 \lambda^2}}{\lambda \alpha \sqrt{1+\lambda^2}}\right)}\right)$ rounds and satisfies $\sqrt{\frac{(1 + \lambda^2)(1 +\lambda^2 +\alpha^2 \lambda^2)}{1+ \lambda^2- \alpha^2 \lambda^4}}$-\MEB validity. 
\end{theorem}
\begin{proof}
    We first show the convergence. In Lemma~\ref{lem:asynch-contraction-4t}, we showed that the radius of the minimum enclosing ball shrinks  $R_{r+1}\leq \lambda\alpha \sqrt{\frac{1+ \lambda^2}{1+\lambda^2 + \alpha^2 \lambda^2}}R_r$ with $1\leq \alpha < \frac{\sqrt{1 +\lambda^2}}{\lambda^2}$. The argument of Theorem~\ref{thm:synch-no-inflation} applies verbatim with contraction factor $\lambda\alpha \sqrt{\frac{1+ \lambda^2}{1+\lambda^2 + \alpha^2 \lambda^2}}$ and gives convergence after $\left\lceil \frac{\log(2R_{\max} / \varepsilon)}
{\log\left(\frac{\sqrt{1+\lambda^2+\alpha^2\lambda^2}}{\lambda \alpha \sqrt{1 +\lambda^2}}\right)} \right \rceil$ rounds. This proves convergence and termination.

Next, we show that the Inflated Asynchronous Adaptive \MEB Contraction Algorithm satisfies $\sqrt{\frac{(1 + \lambda^2)(1 +\lambda^2 +\alpha^2 \lambda^2)}{1+ \lambda^2- \alpha^2 \lambda^4}}$-\MEB validity.
As in Theorem~\ref{thm:synch-no-inflation}, we bound the drift of the center and the contraction of the radius simultaneously and establish
\begin{equation}\label{eq:asynch-center-drift-inflated}
  \|C_{r+1}-C_r\|+\gamma R_{r+1}\leq \gamma R_r
\end{equation}
for a suitable $\gamma \geq 1$. Once this holds, the induction of Theorem~\ref{thm:synch-no-inflation} gives $H^{(r)}\subseteq B(C_0,\gamma R_0)$ for every round $r$. From the proof of Lemma~\ref{lem:asynch-contraction-4t}, every correct output lies in $B(a_r,\sigma_r)$, with 
$\sigma_r^2 = \lambda^2 \alpha^2 \frac{1 +\lambda^2}{1+\lambda^2+\alpha^2\lambda^2}R_r^2 - \frac{\alpha^2 \lambda^2(1+\lambda^2)}{(1 +\lambda^2+ \alpha^2 \lambda^2)^2} \|C_r- p_r\|^2$.
Hence, all new correct values lie inside the ball $H^{(r+1)}\subseteq B(a_r,\sigma_r)$, and the $\MEB$ containment inequality gives $\|C_{r+1}-a_r\|^{2}+R_{r+1}^{2}\leq \sigma_r^{2}$. 
We now bound the validity factor $\gamma$. By the triangle inequality, we get $\|C_{r+1}-C_r\| + \gamma R_{r+1} \leq \|a_r - C_r\| + \|C_{r+1}-a_r \| + \gamma R_{r+1}$. Using the definition of $a_r$ for the first term on the right side, using Cauchy-Schwarz and definition of $\sigma_r$ for the second and third term gives:
\begin{align}
    \|C_{r+1}-C_r\| +\gamma R_{r+1} 
    &\leq \frac{(1+\lambda^{2}) \|p_r -C_r\|}{1 +\lambda^{2}  
    + \alpha^{2}\lambda^{2}}\\
    &+\sqrt{1+\gamma^{2}} \sqrt{\lambda^{2} \alpha^{2} \frac{1+\lambda^{2}}{1+\lambda^{2} +\alpha^{2} \lambda^{2}} R_r^{2}-\frac{\alpha^{2} \lambda^{2} (1+\lambda^{2})}{(1 + \lambda^{2} + \alpha^{2} \lambda^{2})^{2}} \|p_r-C_r\|^{2}}.
\end{align}
Hence, it is enough to choose a value of $\gamma$ such that $\sqrt{\frac{1+\lambda^{2}}{\alpha^{2} \lambda^{2}} +1 +\gamma^{2}} \sqrt{\lambda^{2} \alpha^{2} \frac{1 +\lambda^{2}}{1 + \lambda^{2} + \alpha^{2} \lambda^{2}}} \leq \gamma$. After solving this inequality, the smallest possible choice of $\gamma$ is $\gamma=\sqrt{\frac{(1+\lambda^2)(1+\lambda^2 +\alpha^2 \lambda^2)}{1+\lambda^2 -\alpha^2 \lambda^4}}$. Note that the denominator is positive because $\alpha < \frac{\sqrt{1 +\lambda^2}}{\lambda^2}$. Next we plug in $\gamma$ into Inequality~\ref{eq:asynch-center-drift-inflated} and perform the same analysis as in the proof of Theorem~\ref{thm:synch-no-inflation}. 
We get $H^{(r)} \subseteq B(C_0, \gamma R_0)$. Thus, the Inflated Asynchronous Adaptive \MEB Contraction Algorithm satisfies $\sqrt{\frac{(1+\lambda^2)(1+\lambda^2 +\alpha^2 \lambda^2)}{1+\lambda^2 -\alpha^2 \lambda^4}}$-\MEB validity.

\end{proof}

\begin{corollary}
\label{cor:async-4t-alpha-one}
For $n>4t$, we have $\beta \geq 3$ hence $\lambda \leq \sqrt{3/2}$, so the admissible range for $\alpha$ is $1\leq \alpha < \frac{\sqrt{10}}{3} \approx 1.054$. 
For $\alpha=1$, the Inflated Asynchronous Adaptive \MEB Contraction Algorithm  has contraction factor $\frac{\sqrt{15}}{4} \approx 0.968$ and satisfies $2\sqrt{10} \approx 6.32$-\MEB validity.
\end{corollary}
The range of admissible range for $\alpha$ is for $n>4t$ narrow, so $\alpha$ can be chosen only slightly above one. As fraction $\frac{n}{t}$ grows, $\beta$ increases and $\lambda$ approaches one, so the admissible range broadens to $1\leq \alpha <\sqrt{2}$. 

\begin{remark}
    Substituting $\lambda=1$ into the contraction factor of Lemma~\ref{lem:asynch-contraction-4t} and into the validity constant of Theorem~\ref{thm:asynch-inflation} recovers $\alpha\sqrt{\frac{2}{2 +\alpha^{2}}}$ and $\sqrt{\frac{2(\alpha^{2}+ 2)}{2- \alpha^{2}}}$. The bound $\alpha< \frac{\sqrt{1+ \lambda^{2}}}{\lambda^{2}}$ becomes $\alpha< \sqrt2$, as in Lemma~\ref{lem:asynch-contraction-d+2t} and Theorem~\ref{thm:async-d2t}. In contrast to the synchronous case, the admissible range for $\alpha$ remains bounded even at $\lambda=1$. This restriction $\alpha< \sqrt2$ is imposed by the Gather protocol rather than by ball inflation.
\end{remark}

\section{Discussion and Future Work}\label{sec:discussion}

In this section, we first show that the contraction analysis of the Adaptive \MEB Contraction Algorithm is tight. For this, we give an example in which the contraction factor $1/\sqrt{2}$ for $\alpha=1$ is achieved exactly. We then turn to Minimum-Diameter Averaging (MDA), another multidimensional approximate agreement algorithm with dimension-free contraction and strong resilience guarantees. Since MDA is known not to satisfy convex nor box validity, we derive the \MEB-validity guarantees implied by its known diameter contraction bounds. This allows us to compare MDA with our algorithms under a common validity notion. Finally, we discuss the computational aspects of our approach.

\paragraph{Tightness of the contraction analysis.} 
The factor $\frac{1}{\sqrt{2}}$ in Lemma~\ref{lem:new-contraction} is tight for the geometric analysis of Adaptive \MEB Contraction with $\alpha=1$. 

Consider one round with $n=7$, $t=2$, and $d=2$. Let the five correct values be $H^{(r)}=\{(-1,0),(0,0),(0,0),(0,0),(1,0)\}$.
Thus, $\MEB(H^{(r)})=B((0,0),1)$, so $C_r=(0,0)$ and radius $R_r=1$. Let the two Byzantine values be $b^+=(1,1)$ and $b^-=(-1,-1)$.
Suppose that correct processes with values $(0,0),(1,0)$ have local view $P^+=H^{(r)}\cup\{b^+\}$, 
while correct process with value $(-1,0)$ has local view
$P^-=H^{(r)}\cup\{b^-\}$.
Since $n-t=5$, the local $\MEB$-safe area is the intersection of the \MEBS of all subsets of size $n-t=5$ of the corresponding local view.

We first consider correct processes with local view $P^+$. The distinct candidate balls are $B_0=B((0,0),1)$
coming from the subset containing only correct values, then $B_1=B\left(\left(\frac12,\frac12\right),\frac{1}{\sqrt{2}}\right)$,
coming from the subset $\{(0,0),(1,0),(1,1)\}$ and $B_2=B\left(\left(0,\frac12\right),\frac{\sqrt{5}}{2}\right)$
coming from subsets that contain $\{(-1,0), (0,0), (1,1)\}$. Hence $\safeMEB^+ = B_0\cap B_1\cap B_2$.
The three points $(0,0)$, $(1,0)$, and $(0,1)$ all belong to $\safeMEB^+$. Moreover, $\safeMEB^+\subseteq B_1$, and these three points have minimum enclosing ball exactly $B_1$. Therefore, $\MEB(\safeMEB^+) = B\left(\left(\frac12,\frac12\right),\frac{1}{\sqrt{2}}\right)$.
Thus, for $\alpha=1$, Adaptive \MEB Contraction Algorithm outputs $c^+=\left(\frac12,\frac12\right)$.

By symmetry, a process with local view $P^-$ computes a safe area whose minimum enclosing ball is $B\left(\left(-\frac12,-\frac12\right),\frac{1}{\sqrt{2}}\right)$,
and outputs $c^-=\left(-\frac12,-\frac12\right)$.
Hence the new correct values contain $c^+$ and $c^-$, whose distance is $\|c^+-c^-\|=\sqrt{2}$.
Consequently, the new correct $\MEB$ has radius $R_{r+1}=\frac{\sqrt{2}}{2}=\frac{1}{\sqrt{2}}R_r$. 
Thus, the contraction factor $1/\sqrt{2}$ is achieved. The example is illustrated in Figure~\ref{fig:tight-example-contraction}.

\DeclareRobustCommand{\correctmarker}{%
\tikz[baseline=-0.55ex]\node[circle,fill=black,inner sep=1.6pt]{};%
}
\DeclareRobustCommand{\byzmarker}{%
\tikz[baseline=-0.55ex]\node[rectangle,fill=red!70,draw=red!70,minimum size=5pt,inner sep=0pt]{};%
}
\DeclareRobustCommand{\supportmarker}{%
\tikz[baseline=-0.55ex]\node[circle,draw=blue!70,fill=blue!20,inner sep=1.5pt]{};%
}
\DeclareRobustCommand{\outputmarker}{%
\tikz[baseline=-0.55ex]\node[circle,fill=blue!75,inner sep=1.8pt]{};%
}

\begin{figure}[t]
\centering
\begin{tikzpicture}[scale=1.15,>=latex]

\tikzset{
    correct/.style={circle,fill=black,inner sep=1.6pt},
    byz/.style={rectangle,fill=red!70,draw=red!70,minimum size=5pt,inner sep=0pt},
    safe/.style={fill=blue!15,draw=blue!60,thick},
    cand/.style={draw=gray!70,dashed},
    output/.style={circle,fill=blue!75,inner sep=1.8pt},
    support/.style={circle,draw=blue!70,fill=blue!20,inner sep=1.5pt},
    lab/.style={font=\scriptsize}
}

% =========================
% Right local view
% =========================
\begin{scope}[xshift=-2.9cm]
    \node at (0,2.25) {\textbf{Right local view}};

    % axes
    \draw[->,gray!60] (-1.55,0) -- (1.75,0);
    \draw[->,gray!60] (0,-0.35) -- (0,1.7);

    % safe area: intersection of the three candidate balls
    \begin{scope}
        \clip (0,0) circle[radius=1];
        \clip (0.5,0.5) circle[radius={sqrt(0.5)}];
        \filldraw[safe] (0,0.5) circle[radius={sqrt(5)/2}];
    \end{scope}

    % candidate balls
    \draw[cand] (0,0) circle[radius=1];
    \draw[cand] (0.5,0.5) circle[radius={sqrt(0.5)}];
    \draw[cand] (0,0.5) circle[radius={sqrt(5)/2}];

    % points
    \node[correct,label={[lab,xshift=-5pt,yshift=-1pt]below right:{$(0,0)$}}] at (0,0) {};
    \node[correct,label={[lab]below right:{$(1,0)$}}] at (1,0) {};
    \node[correct,label={[lab]below:{$(-1,0)$}}] at (-1,0) {};
    \node[byz,label={[lab]above right:{$(1,1)$}}] at (1,1) {};

    % boundary/support point of the safe area
    \node[support,label={[lab]above left:{$(0,1)$}}] at (0,1) {};

    % output
    \node[output,label={[lab]right:{$c_R=(\frac12,\frac12)$}}] at (0.5,0.5) {};
\end{scope}

% =========================
% Left local view
% =========================
\begin{scope}[xshift=2.9cm]
    \node at (0,2.25) {\textbf{Left local view}};

    % axes
    \draw[->,gray!60] (-1.75,0) -- (1.55,0);
    \draw[->,gray!60] (0,-1.7) -- (0,0.35);

    % safe area: symmetric intersection of the three candidate balls
    \begin{scope}
        \clip (0,0) circle[radius=1];
        \clip (-0.5,-0.5) circle[radius={sqrt(0.5)}];
        \filldraw[safe] (0,-0.5) circle[radius={sqrt(5)/2}];
    \end{scope}

    % candidate balls
    \draw[cand] (0,0) circle[radius=1];
    \draw[cand] (-0.5,-0.5) circle[radius={sqrt(0.5)}];
    \draw[cand] (0,-0.5) circle[radius={sqrt(5)/2}];

    % points
    \node[correct,label={[lab]above right:{$(0,0)$}}] at (0,0) {};
    \node[correct,label={[lab]below left:{$(-1,0)$}}] at (-1,0) {};
    \node[correct,label={[lab]below:{$(1,0)$}}] at (1,0) {};
    \node[byz,label={[lab]below left:{$(-1,-1)$}}] at (-1,-1) {};

    % boundary/support point of the safe area
    \node[support,label={[lab]below right:{$(0,-1)$}}] at (0,-1) {};

    % output
    \node[output,label={[lab,yshift=-5pt]left:{$c_L=(-\frac12,-\frac12)$}}] at (-0.5,-0.5) {};
\end{scope}

\end{tikzpicture}
\caption{Tight example showing that the $\alpha=1$ contraction analysis achieves $\frac{1}{\sqrt{2}}$ contraction factor. The right and left local $\MEB$-safe areas are symmetric and produce outputs $c_R=(\frac12,\frac12)$ and $c_L=(-\frac12,-\frac12)$. Hence $\|c_R-c_L\|=\sqrt{2}$, so the new correct $\MEB$ has radius $R_{r+1}=1/\sqrt{2}=R_r/\sqrt{2}$. Correct values are shown as black points \correctmarker, while Byzantine values are shown as red squares \byzmarker. Candidate $\MEB$s are shown with dashed lines and the $\MEB$-safe area is highlighted in blue. The points $(0,1)$ and $(0,-1)$ shown as \supportmarker $ $  are not input values; they are illustrated only to make the boundary of the $\MEB$-safe areas visible.}
\label{fig:tight-example-contraction}
\end{figure}
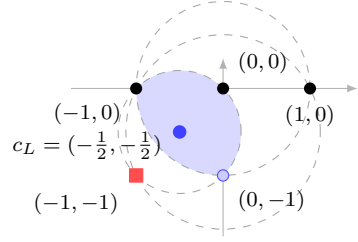

\paragraph{Comparison with Minimum-Diameter Averaging.}
Another well-known approximate agreement algorithm is Minimum-Diameter Averaging (MDA)~\cite{10.5555/3540261.3542179}. MDA is particularly relevant in our setting because, like our algorithms, it provides a dimension-free contraction. However, MDA is analysed in terms of diameter rather than radius contraction and it is known to satisfy strong validity but neither box nor convex validity~\cite{10.1007/978-3-032-11127-2_10}. In this section, we therefore derive the 
\MEB-validity guarantee that follows from MDA's diameter contraction and compare it to our algorithm.

\begin{definition}[MDA]
    Let $P=\{m_1,\dots,m_n\}\subset \mathbb{R}^d$ and let $t$ be the maximal number of Byzantine processes. Define the diameter of a set $M\subseteq P$ as $D_M = \max_{m_i,m_j\in M} \|m_i-m_j\|$.
An MDA output is obtained as follows: choose a subset
$M \in \arg\min_{\substack{M\subseteq P \\ |M|=n-t}} D_M$, and output $ MDA = \frac{1}{|M|} \sum_{m_i\in M} m_i$.
\end{definition}

MDA selects a minimum-diameter subset of received values and outputs their average. 
Cambus and Melnyk~\cite{10.1007/978-3-032-11127-2_10} show that in the synchronous case, MDA tolerates up to $n>4t$ Byzantine processes and has contraction rate $D_{r+1}\le \frac{2}{3}D_r$, where $D_r$ denotes the diameter of correct processes in round $r$. 
In the asynchronous case, authors in \cite{10.5555/3540261.3542179}, show that MDA has $n>7t$ resilience and contraction factor $D_{r+1}\le \frac{4}{5}D_r$. 
There, a process does not apply MDA to all $n$ values. Instead, it applies MDA to its local view $P_i^{(r)}$, where $q=|P_i^{(r)}|$, and selects a minimum-diameter subset of size $q-t=|P_i^{(r)}|-t$.

\begin{lemma}\label{lem:mda-validity}
    In the synchronous setting with $n>4t$, MDA satisfies $7$-\MEB validity. In the asynchronous setting of El-Mhamdi et al.~\cite{10.5555/3540261.3542179}, where process $i$ applies MDA to its local view $P_i^{(r)}$ and selects a subset of size $|P_i^{(r)}|-t$, MDA satisfies $11$-\MEB validity under the contraction bound $D_{r+1}\leq \frac{4}{5} D_r$.
\end{lemma}
\begin{proof}
    Let $D_r$ denote the diameter of the correct values in round $r$, let $\MEB(H^{(0)})=B(C_0, R_0)$, and let $L_r=\max_{x\in H^{(r)}}\|x-C_0\|$ be the smallest radius with $H^{(r)} \subseteq B(C_0,L_r)$, so that $L_0\leq R_0$. In other words, $L_r$ measures how far the current correct values in round $r$ have drifted from the correct \MEB's center $C_0$.
    Let $q_{\mathrm{MDA}}<1$ denote the contraction factor of MDA, i.e., $D_{r+1}\le q_{\mathrm{MDA}}D_r$.
In the synchronous setting, $q_{\mathrm{MDA}}=\frac{2}{3}$~\cite{10.1007/978-3-032-11127-2_10}, whereas in the asynchronous setting~\cite{10.5555/3540261.3542179}, $q_{\mathrm{MDA}}=\frac{4}{5}$.
    
    Fix a round $r$ and consider the output of a correct process $i$. Let $P_i^{(r)}$ be the set of values used by process $i$ in this round, and let $M$ be the subset selected by MDA. In the synchronous setting, $P_i^{(r)}$ contains all $n$ values and MDA selects a subset of size $n-t$. In the asynchronous setting of El-Mhamdi et al.~\cite{10.5555/3540261.3542179}, process $i$ applies MDA to its local view $P_i^{(r)}$ and selects a subset of size $|P_i^{(r)}| -t$. Let $D_M$ be the diameter of the subset $M$. 
    
    In both settings, obtaining a local subset containing only correct values is feasible.
    In the synchronous model, due to consistent broadcast, it is guaranteed that all correct processes receive all correct values. Hence, a correct process locally computes a subset of size $n-t$ consisting of correct values only. In the asynchronous setting, at most $t$ values in $P_i^{(r)}$ are Byzantine, so $P_i^{(r)}$ contains at least $|P_i^{(r)}|-t$ correct values. So obtaining a subset of correct values of size $|P_i^{(r)}|-t$ is feasible. Since every subset containing only correct values has diameter at most $D_r$ and the set $M$ has minimum diameter, we have $D_M \leq  D_r$.
    Moreover, $M$ contains at least one correct value, since the selected subset has size larger than $t$ and at most $t$ values are Byzantine.

    Consider one such correct value and name it $h$. Since $D_M$ is the diameter of $M$, every value in $M$ lies within distance at most $D_M$ from $h$. So, the subset $M$ lies inside the ball centered at $h$ with radius $D_M$, i.e. $M\subseteq B(h,D_M)$. As $B(h,D_M)$ is convex and the MDA output is the average of the values in $M$, the output lies in $B(h,D_M)$ as well, hence at distance at most $D_r$ from $h$. 
    
    Since $h\in H^{(r)} \subseteq B(C_0,L_r)$, the triangle inequality implies that every correct value in round $r+1$ lies in $B(C_0,L_r+D_r)$. Hence, $L_{r+1}\leq L_r+D_r$. Iterating over rounds $0, \dots,r-1$ and using the contraction $D_k\leq q_{\mathrm{MDA}}^k D_0$ we obtain:
    $$L_r \leq L_0+ \sum_{k=0}^{r-1}D_k \leq R_0+D_0 \sum_{k=0}^{\infty}q_{\mathrm{MDA}}^k = R_0+\frac{D_0}{1-q_{\mathrm{MDA}}}. $$
    Since $D_0\leq 2R_0$, we obtain $L_r\leq \left(1+\frac{2}{1-q_{\mathrm{MDA}}}\right)R_0$. 
    
    For synchronous MDA, $q_{\mathrm{MDA}}=\frac{2}{3}$, so MDA satisfies 7-\MEB validity. For asynchronous MDA, $q_{\mathrm{MDA}}=\frac{4}{5}$, so MDA satisfies 11-\MEB validity.
\end{proof}

The validity guarantee is obtained by summing over drift in each round via geometric series. Our bounds instead control the drift of the center and the contraction of the radius simultaneously, as discussed after Theorem~\ref{thm:synch-no-inflation}, and are correspondingly smaller. For the synchronous setting, the bound for Adaptive \MEB Contraction Algorithm is $\sqrt{6}$-\MEB validity with resilience $n>3t$, whereas MDA has resilience $n>4t$ and satisfies 7-\MEB validity. In the asynchronous setting, we achieve $2\sqrt{10}$-\MEB validity with resilience $n>4t$. In contrast, MDA satisfies 11-\MEB validity and has resilience $n>7t$.

\noindent\textbf{Computational Aspects and Future Work.} 
As is standard in distributed computing and agreement algorithms analysis, our results in Theorems~\ref{thm:synch-no-inflation}, \ref{thm:synch-inflation}, \ref{thm:async-d2t} and~\ref{thm:asynch-inflation} are focused on round complexity rather than on the local computation performed in each round. The number of communication rounds follows directly from the contraction factor and is logarithmic in the ratio between the initial radius and the target distance $\varepsilon$. The local computation in each round can, however, be expensive. 
Every algorithm of this type, such as Mendes--Herlihy and Vaidya--Garg~\cite{mendes2015multidimensional}, pays for computing a safe area and the cost is dominated by the $\binom{n}{t}$ candidate subsets. In our algorithms the safe area region is an intersection of minimum enclosing balls.

Our Adaptive \MEB Contraction algorithm is inspired by the core-set construction of B\u{a}doiu and Clarkson~\cite{badoiu2003smaller,badoiu2008optimal}. In each iteration, the algorithm computes the minimax center of the current selected set $S_i$ and then queries the $\MEB$-safe area for a point farthest from this center. The difference from the standard core-set setting is that this farthest point is not chosen from a finite input set, but from the continuous region $\safeMEB_i^{(r)}$. Thus, an exact implementation requires optimization over the safe area. The parameter $\alpha$ controls how often this query is needed. By Lemma~\ref{lem:midextremes}, for $\alpha=\sqrt3$ the initial diameter pair already satisfies the stopping criterion, so no additional queries for the farthest point is required. Smaller values of $\alpha$ improve the contraction factor, however they may require more iterations. Thus, $\alpha$ gives a trade off between local computation and round complexity.

The inflated variants of the Adaptive \MEB Contraction algorithm have an additional computational overhead. To obtain the containment of the \MEBS needed for the analysis, we intersect over all subsets of size at least $n-t$ in the synchronous case, and at least $|P_i^{(r)}|-t$ in the asynchronous case. This enlarges the family of candidate balls and can make a direct implementation substantially more expensive. 
The question whether the intersection of the minimum enclosing balls over all
subsets of size exactly $n-t$ consisting only of correct values is already contained
in $\MEB(H^{(r)})$ remains open.
Our goal in this work was to achieve better resilience bounds and we have not optimized the resulting computation. Designing efficient implementations and approximation routines for the inflated $\MEB$-safe areas remains an important direction for future work.

\section*{Acknowledgments}
Research supported by the German Research Foundation (DFG), Schwerpunktprogramm SPP 2378: Resilience in Connected Worlds: Mastering Failures, Overload, Attacks, and the Unexpected, ReNO-2 (511099228), 2025-2029.

\section*{AI Disclosure}
We used ChatGPT and Claude (Anthropic) to assist with the written presentation and clarity of the paper. All technical results, definitions, algorithms and proofs originate from the authors, who verified the correctness and originality of all content including references.

\bibliography{literature}
\end{document}